\documentclass[submission,copyright,creativecommons]{eptcs}
\providecommand{\event}{QPL 2021} % Name of the event you are submitting to
\usepackage{breakurl}             % Not needed if you use pdflatex only.
\usepackage{underscore}           % Only needed if you use pdflatex.

\input{realstab.sty}

\title{Generators and Relations for Real Stabilizer Operators}
\author{Justin Makary, Neil J. Ross, and Peter Selinger
\institute{Department of Mathematics and Statistics}
\institute{Dalhousie University, Halifax, Canada}
}
\def\titlerunning{Generators and Relations for Real Stabilizer Operators}
\def\authorrunning{Justin Makary, Neil J. Ross \& Peter Selinger}
\begin{document}
\maketitle

\begin{abstract}
Real stabilizer operators, which are also known as real Clifford operators, are generated, through composition and tensor product, by the Hadamard gate, the Pauli $Z$ gate, and the controlled-$Z$ gate. We introduce a normal form for real stabilizer circuits and show that every real stabilizer operator admits a unique normal form. Moreover, we give a finite set of relations that suffice to rewrite any real stabilizer circuit to its normal form.
\end{abstract}

\section{Introduction}
\label{sec:intro}

\emph{Stabilizer} operators, which are also known as \emph{Clifford}
operators, play a fundamental role in the study of fault-tolerant
quantum computation \cite{NielsenChuang}. Stabilizer operators are
generated, under composition and tensor product, by the scalar
$e^{i\pi/4}$, the Hadamard gate $H$, the phase gate $S$, and the
controlled-$Z$ gate $CZ$. For $n \geq 0$, the set of stabilizer
operators on $n$ qubits forms a subgroup of the unitary group
$\Uni{2^n}$ and is denoted $\clifford{n,\C}$. Quantum circuits for
stabilizer operators have been extensively studied
\cite{Aaronson2004ImprovedSO,Backens2014TheZI,
  Bravyi2020HadamardfreeCE,Gottesman1998TheHR,calderbank,
  Selinger15,Nest2010ClassicalSO}. In particular, \cite{Selinger15}
gave a finite presentation of stabilizer operators by introducing a
normal form for stabilizer circuits together with a finite collection
of relations that suffice to rewrite any stabilizer circuit to its
normal form.

In the present paper, we study \emph{real} stabilizer operators which
are generated by the scalar $(-1)$ and the gates $H$, $Z$, and
$CZ$. The group of $n$-qubit real stabilizer operators
$\clifford{n,\R}$ is the intersection of $\clifford{n,\C}$ and the
orthogonal group $\Ortho{2^n}$.

Our contributions are as follows. We define a normal form for real
stabilizer circuits and we prove that every real Clifford operator
admits a unique normal form. We then introduce a finite collection of
relations between real stabilizer circuits and show that the relations
are complete. The completeness of the relations is established by
formulating a rewrite system to transform any real stabilizer circuit
into its normal form. Our work follows the methods of
\cite{Selinger15} but the focus on real operators requires a distinct
notion of normal form. In order to conveniently describe these normal
forms, we introduce a \emph{typing} for quantum circuits.

Restrictions such as the one considered here were previously studied
in the context of randomized benchmarking \cite{Hashagen2018RealRB},
graphical languages \cite{BK2019,Vilmart2018AZW}, and exact synthesis
\cite{Amy2020NumberTheoreticCO}. Real stabilizers were explicitly
investigated in \cite{Comfort19,DP3,Hashagen2018RealRB,Nebe00}. In
particular, \cite{Comfort19} provides a complete set of circuit
equalities for real stabilizer circuits with ancillas. The presence of
ancillas, however, implies that the circuits discussed in
\cite{Comfort19} do not always correspond to orthogonal operators. In
contrast, the circuits discussed here always represent orthogonal
operators.

The paper is organized as follows. In \cref{sec:groups}, we examine
the structure of the real Pauli and Clifford groups. In
\cref{sec:circuits}, we review the diagrammatic language of quantum
circuits and introduce annotated and typed circuits. In \cref{sec:nf},
we define normal forms and prove that every real stabilizer operator
admits a unique normal form. We state our relations in \cref{sec:rels}
and propose a system for rewriting any real stabilizer circuit to its
normal form. We discuss future work in \cref{sec:conc}.

\section{The Real Pauli and Clifford Groups}
\label{sec:groups}

We denote the transpose of the matrix $A$ by $A^\intercal$. A matrix
$A$ is symmetric if $A=A^\intercal$ and orthogonal if
$A^{-1}=A^\intercal$. Following \cite{Selinger15}, for two matrices
$A$ and $B$, we write $A\bullet B$ for $ABA^{-1}$. Throughout, we use
the terms ``operator'' and ``matrix'' interchangeably, assuming that
operators are always represented with respect to the standard
(computational) basis.

The Pauli matrices $X$ and $Z$, the Hadamard matrix $H$ and the
controlled-$Z$ matrix $CZ$ are defined as
\[
X= 
\begin{bmatrix} 
  0 & 1 \\
  1 & 0 
\end{bmatrix}, 
\quad
Z= \begin{bmatrix}
  1 & 0 \\
  0 & -1
\end{bmatrix},
\quad
H= \frac{1}{\sqrt{2}}
\begin{bmatrix} 
  1 & 1 \\
  1 & -1 
\end{bmatrix},
\quad \mbox{and} \quad 
CZ = \begin{bmatrix} 
  1 & 0 & 0 & 0 \\
  0 & 1 & 0 & 0 \\
  0 & 0 & 1 & 0 \\
  0 & 0 & 0 & -1 \\
\end{bmatrix}.
\]
We note that $X$ and $Z$ are orthogonal and symmetric so that $X^{2} =
Z^{2} = I$. Moreover, $X$ and $Z$ anticommute: $XZ = -ZX$. This
implies that $(XZ)^{2} = -1$ so that $XZ$ is orthogonal but not
symmetric.

\begin{definition}
  \label{def:pauligroup}
  The \emph{real Pauli group on $n$ qubits} $\pauli{n,\R}$ is defined
  as
  \[
  \pauli{n,\R}=\s{\pm (P_1\otimes \ldots \otimes P_n) \mid P_i \in
    \s{I, X, Z, XZ}}.
  \]
\end{definition}  

In what follows, we drop the adjective ``real'' and simply refer to
$\pauli{n,\R}$ as the Pauli group. In addition, we write $\pauli{n}$
for $\pauli{n,\R}$. We note that the $n$-qubit Pauli group $\pauli{n}$
spans the vector space of real $2^n\times 2^n$ matrices. The
proposition below records an important property of Pauli operators.

\begin{proposition}
  \label{prop:psquare}
  Let $P=(-1)^a(P_1\otimes\ldots\otimes P_n)$ with $P_i\in\s{I, X, Z,
    XZ}$. Then $P^2=I$ if and only if there are evenly many $i$ such
  that $P_i=XZ$.
\end{proposition}

\begin{proof}
If $P_i\in\s{I, X, Z}$ then $P_i^2=I$ and if $P_i=XZ$ then
$P_i^2=-1$. Hence, for any $P\in\pauli{n}$, $P_i^2=(-1)^dI$ where $d$
is the number of components for which $P_i=XZ$. Thus, $P^2=I$ if and
only if $d$ is even.
\end{proof}

\begin{definition}
  \label{def:clifford}
  The \emph{real Clifford group on $n$ qubits} $\clifford{n,\R}$ is
  the normalizer of $\pauli{n}$ in $\Ortho{2^n}$. That is,
  \[
    \clifford{n,\R} = \{U \in \Ortho{2^n} \ | \ U \bullet P \in
    \pauli{n} \mbox{ for all } P \in \pauli{n}\}.
  \]
\end{definition}

As with the Pauli group, we drop the adjective ``real'' when referring
to $\clifford{n,\R}$ in what follows and, for brevity, write
$\clifford{n}$ for $\clifford{n,\R}$. Since the Clifford group is the
normalizer of the Pauli group, we have that $C \bullet P \in
\pauli{n}$ for every Clifford $C$ and every Pauli $P$. Furthermore,
conjugation is a group automorphism of $\pauli{n}$. Note that $H \in
\clifford{1}$, $CZ \in \clifford{2}$, and $\pauli{n} \subseteq
\clifford{n}$.

\begin{proposition}
  \label{prop:scalar}
  Let $C \in \clifford{n}$. If $C \bullet P = P$ for all $P \in
  \pauli{n}$, then $C= \pm 1$.
\end{proposition}

\begin{proof}
By assumption, $CPC^{-1} = P$, for all $P \in \pauli{n}$. Since
$\pauli{n}$ spans the space of $2^{n} \times 2^{n}$ real matrices, it
follows that for any $2^{n} \times 2^{n}$ operator $N$, $CNC^{-1} =
N$. Thus, $C$ commutes with every real matrix and is therefore a
scalar. Because the only scalars in $\Ortho{2^n}$ are $\pm 1$, we get
$C=\pm 1$.
\end{proof}

\begin{corollary}
If $C$ and $D$ are two elements of $\clifford{n}$ that act identically
on $\pauli{n}$, then $C=\pm D$.
\end{corollary}

\begin{proof}
Since $C$ and $D$ act identically on $\pauli{n}$, we have $(D^{-1}C)
\bullet P = D^{-1} \bullet C \bullet P = D^{-1}\bullet D \bullet P =
P$. Thus, by \cref{prop:scalar}, $D^{-1}C = \pm 1$. Hence $C = \pm D$.
\end{proof}

% --------------------------------------------------------------------
\section{Annotated and Typed Circuits}
\label{sec:circuits}

We assume that the reader is familiar with the language of quantum
circuits \cite{NielsenChuang}. In this section, we introduce certain
decorations which will be convenient in discussing circuits.

The Hadamard gate, the Pauli $Z$ gate, and the controlled-$Z$ gates
are represented below.
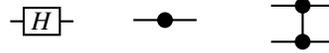
\[
\raisebox{-0.4em}{
\begin{qcircuit}[scale=0.3]
  \blackgrid{0}{2.8}{0}
  \HGATE{1.4}{0}
\end{qcircuit}}
\qquad
\begin{qcircuit}[scale=0.3]
  \blackgrid{0}{2.8}{-0.5}
  \ZGATE{1.4}{-0.5}
\end{qcircuit}
\qquad
\raisebox{-0.75em}{
\begin{qcircuit}[scale=0.3]
    \blackgrid{0}{2.8}{0,1.6}
    \CZGATE{1.4}{1.6}{0}
\end{qcircuit}}
\]
For brevity, we introduce some \emph{derived} gates which are
shorthand for certain Clifford circuits.
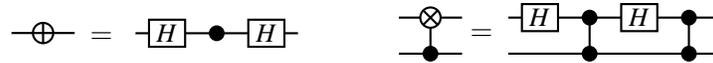
\[
\begin{array}{ll}
\begin{qcircuit}[scale=0.3]
    \blackgrid{0}{2.8}{0}
    \XGATE{1.4}{0}
\end{qcircuit}\ \
\raisebox{0.05em}{$=$}\
\raisebox{-0.25em}{
\begin{qcircuit}[scale=0.3]
    \blackgrid{0}{7.2}{0}
    \HGATE{1.4}{0}
    \ZGATE{3.6}{0}
    \HGATE{5.8}{0}
\end{qcircuit}}
& \ \ \ \ \ \ \ \ 
\raisebox{-0.6em}{
\begin{qcircuit}[scale=0.3]
    \blackgrid{0}{2.8}{0,1.6}
    \XZCGATE{1.4}{1.6}{0}
\end{qcircuit}
\raisebox{0.75em}{$=$}
\begin{qcircuit}[scale=0.3]
    \blackgrid{0}{9.4}{0,1.6}
    \HGATE{1.4}{1.6}
    \CZGATE{3.6}{1.6}{0}
    \HGATE{5.8}{1.6}
    \CZGATE{8}{1.6}{0}
\end{qcircuit}}
\end{array}
\]
The derived gate on the left is the Pauli $X$ gate. We call the
derived gate on the right the $CXZ $ gate, an abbreviation for
\emph{controlled-$XZ$}.

We introduce \emph{annotations} on circuits to concisely indicate the
action of a Clifford operator on a Pauli operator under
conjugation. When $C \in \clifford{n}$, and $P = P_{1} \otimes \cdots
\otimes P_{n}$, $Q = Q_{1} \otimes \cdots \otimes Q_{n} \in
\pauli{n}$, we write
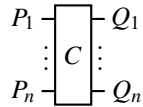
\[
\begin{qcircuit}[scale=0.3]
  \LABEL{0.2}{1.8}{$\vdots$}
  \LABEL{2.6}{1.8}{$\vdots$}
  \LABEL{3.8}{0}{$Q_{n}$}
  \LABEL{3.75}{3.2}{$Q_{1}$}
  \LABEL{-0.8}{0}{$P_{n}$}
  \LABEL{-0.8}{3.2}{$P_{1}$}    
  \blackgrid{0}{2.8}{0,3.2}
  \GENERIC{1.4}{3.2}{0}{$C$}
\end{qcircuit}
\]
to indicate that $C \bullet P = Q$.

It will be useful for our purposes to \emph{type} circuits. The notion
of a \emph{typed gate} coincides with the usual notion of gate, with
the difference that the wires of the gate can be of \emph{simple type}
or of \emph{double type}, as shown in the two examples below.
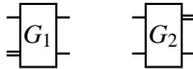
\[
\begin{array}{lll}
\begin{qcircuit}[scale=0.3]
    \redgrid{0}{1.4}{0}
    \blackgrid{1.4}{2.8}{0}
    \blackgrid{0}{1.4}{1.6}
    \greengrid{1.4}{2.8}{1.6}
    \GENERIC{1.4}{1.6}{0}{$G_{1}$}
\end{qcircuit}
& ~ &
\begin{qcircuit}[scale=0.3]
    \blackgrid{0}{1.4}{0}
    \blackgrid{1.4}{2.8}{0}
    \greengrid{0}{1.4}{1.6}
    \redgrid{1.4}{2.8}{1.6}
    \GENERIC{1.4}{1.6}{0}{$G_{2}$}
\end{qcircuit}
\end{array}
\]
The type of wires does not affect the vertical composition of gates,
but two gates can only be composed horizontally if the types of the
corresponding wires are the same. For example, below, the composition
on the left is well-defined but the composition on the right is not.
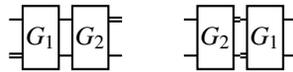
\[
\begin{qcircuit}[scale=0.3]
    \redgrid{0}{1.4}{0}
    \blackgrid{0}{1.4}{1.6}
    \greengrid{1.4}{3.6}{1.6}
    \blackgrid{1.4}{3.6}{0}
    \redgrid{3.6}{5}{1.6}
    \blackgrid{3.6}{5}{0}
    \GENERIC{1.4}{1.6}{0}{$G_{1}$}
    \GENERIC{3.6}{1.6}{0}{$G_{2}$}
\end{qcircuit}
\qquad
\begin{qcircuit}[scale=0.3]
    \blackgrid{0}{1.4}{0}
    \greengrid{0}{1.4}{1.6}
    \redgrid{1.4}{2.5}{1.6}
    \blackgrid{2.5}{3.6}{1.6}    
    \blackgrid{1.4}{2.5}{0}
    \redgrid{2.5}{3.6}{0}    
    \greengrid{3.6}{5}{1.6}
    \blackgrid{3.6}{5}{0}
    \GENERIC{1.4}{1.6}{0}{$G_{2}$}
    \GENERIC{3.6}{1.6}{0}{$G_{1}$}
\end{qcircuit}
\]
\emph{Typed circuits} are constructed from typed gates with this
restriction. The typing of gates and circuits is meant to constrain
the construction of circuits.

A typed gate is defined in two stages. In the first stage, a standard
gate is specified, for example by associating a diagram to a matrix or
to a circuit made from preexisting gates. In the second stage, types
are associated to the input and output wires of the gate. Note that
any typed circuit can still be viewed as an un-typed circuit by simply
\emph{forgetting} about the types of the wires.

By abuse of notation, we will sometimes use a single circuit to
concisely specify a family of (typed) circuits. As an illustration,
consider the typed gates below.
\[
\begin{array}{llll}
\begin{qcircuit}[scale=0.3]
    \blackgrid{0}{1.4}{0}
    \greengrid{1.4}{2.8}{0}
    \GENERIC{1.4}{0}{0}{$G_{1}$}
\end{qcircuit}
&
\begin{qcircuit}[scale=0.3]
    \blackgrid{0}{1.4}{0}
    \redgrid{1.4}{2.8}{0}
    \GENERIC{1.4}{0}{0}{$G_{2}$}
\end{qcircuit}
&
\begin{qcircuit}[scale=0.3]
    \redgrid{0}{1.4}{0}
    \blackgrid{1.4}{2.8}{0}
    \GENERIC{1.4}{0}{0}{$H_{1}$}
\end{qcircuit}
&
\begin{qcircuit}[scale=0.3]
    \greengrid{0}{1.4}{0}
    \blackgrid{1.4}{2.8}{0}
    \GENERIC{1.4}{0}{0}{$H_{2}$}
\end{qcircuit}
\end{array}
\]
Then the diagram
\[
\begin{qcircuit}[scale=0.3]
    \blackgrid{0}{1.4}{0}
    \blackgrid{1.4}{3.6}{0}
    \blackgrid{3.6}{5}{0}
    \GENERIC{1.4}{0}{0}{$G$}
    \GENERIC{3.6}{0}{0}{$H$}
\end{qcircuit}
\]
represents the family of circuits in which the gate on the left-hand
side is one of $G_1$ or $G_2$ and the gate on the right-hand side is
one of $H_1$ or $H_2$ subject to the condition that the circuit is a
well-formed typed circuit. There are two circuits in this specific
family, which are represented below.
\[
\begin{array}{lll}
\begin{qcircuit}[scale=0.3]
    \blackgrid{0}{1.4}{0}
    \greengrid{1.4}{3.6}{0}
    \blackgrid{3.6}{5}{0}
    \GENERIC{1.4}{0}{0}{$G_1$}
    \GENERIC{3.6}{0}{0}{$H_2$}
\end{qcircuit}
& ~ &
\begin{qcircuit}[scale=0.3]
    \blackgrid{0}{1.4}{0}
    \redgrid{1.4}{3.6}{0}
    \blackgrid{3.6}{5}{0}
    \GENERIC{1.4}{0}{0}{$G_2$}
    \GENERIC{3.6}{0}{0}{$H_1$}
\end{qcircuit}
\end{array}
\]

\section{Normal Forms for Real Stabilizer Circuits}
\label{sec:nf}

We now introduce \emph{normal forms} for stabilizer operators. That
is, we specify a family of circuits and show that every stabilizer
operator is represented by a unique element of this family.

\subsection{Derived Generators}

We start by introducing \emph{derived generators}, which will serve as
the basic building blocks for our normal forms. As discussed in
\cref{sec:circuits}, we introduce these derived generators in two
stages: first we define the gates as (un-typed) circuits and then we
specify the types of their wires. There are five kinds of derived
generators: $A$, $B$, $C$, $D$, and $E$.

\begin{definition}
\label{def:agates}
The \emph{derived generators of kind $A$} are defined below.
\[
\begin{array}{lll}
\begin{qcircuit}[scale=0.3]
    \blackgrid{0}{2.8}{0}
    \AONE{1.4}{0}
\end{qcircuit} \ \
\raisebox{0.5em}{$=$} \
\raisebox{0.64em}{
\begin{qcircuit}[scale=0.3]
    \blackgrid{0}{2.8}{0}
\end{qcircuit}}
& \ \ \ \ \ 
\begin{qcircuit}[scale=0.3]
    \blackgrid{0}{2.8}{0}
    \ATWO{1.4}{0}
\end{qcircuit}\ \
\raisebox{0.3em}{$=$}
\raisebox{0.02em}{
\begin{qcircuit}[scale=0.3]
    \blackgrid{0}{2.8}{0}
    \HGATE{1.4}{0}
\end{qcircuit}}
& \ \ \ \ \ 
\begin{qcircuit}[scale=0.3]
    \blackgrid{0}{2.8}{0}
    \ATHREE{1.4}{0}
\end{qcircuit}\ \
\raisebox{0.3em}{$=$}
\raisebox{0.64em}{
\begin{qcircuit}[scale=0.3]
    \blackgrid{0}{2.8}{0}
\end{qcircuit}}
\end{array}
\]
\end{definition}

\begin{definition}
\label{def:bgates}
The \emph{derived generators kind $B$} are defined below.
\[
\begin{array}{l}
\begin{qcircuit}[scale=0.3]
    \blackgrid{0}{2.8}{0}
    \blackgrid{0}{2.8}{1.6}
    \BONE{1.4}{1.6}{0}
\end{qcircuit}
\raisebox{1em}{$=$}
\begin{qcircuit}[scale=0.3]
    \blackgrid{0}{2.8}{0}
    \blackgrid{0}{2.8}{1.6}
    \BFIVE{1.4}{1.6}{0}
\end{qcircuit}
\raisebox{1em}{$=$}
\begin{qcircuit}[scale=0.3]
    \blackgrid{0}{13.8}{0,1.6}
    \HGATE{1.4}{0}
    \CZGATE{3.6}{1.6}{0}
    \HGATE{5.8}{0}
    \HGATE{5.8}{1.6}
    \CZGATE{8}{1.6}{0}
    \HGATE{10.2}{0}
    \HGATE{10.2}{1.6}
    \CZGATE{12.4}{1.6}{0}
\end{qcircuit}
\\
\begin{qcircuit}[scale=0.3]
    \blackgrid{0}{2.8}{0}
    \blackgrid{0}{2.8}{1.6}
    \BTWO{1.4}{1.6}{0}
\end{qcircuit}
\raisebox{1em}{$=$}
\begin{qcircuit}[scale=0.3]
    \blackgrid{0}{2.8}{0}
    \blackgrid{0}{2.8}{1.6}
    \BSIX{1.4}{1.6}{0}
\end{qcircuit}
\raisebox{1em}{$=$}
\begin{qcircuit}[scale=0.3]
    \blackgrid{0}{7.2}{0,1.6}
    \CZGATE{1.4}{1.6}{0}
    \HGATE{3.6}{0}
    \HGATE{3.6}{1.6}
    \CZGATE{5.8}{1.6}{0}
\end{qcircuit}
\\
\begin{qcircuit}[scale=0.3]
    \blackgrid{0}{2.8}{0}
    \blackgrid{0}{2.8}{1.6}
    \BTHREE{1.4}{1.6}{0}
\end{qcircuit}
\raisebox{1em}{$=$}
\begin{qcircuit}[scale=0.3]
    \blackgrid{0}{2.8}{0}
    \blackgrid{0}{2.8}{1.6}
    \BSEV{1.4}{1.6}{0}
\end{qcircuit}
\raisebox{1em}{$=$}
\begin{qcircuit}[scale=0.3]
    \blackgrid{0}{9.4}{0,1.6}
    \HGATE{1.4}{1.6}
    \CZGATE{3.6}{1.6}{0}
    \HGATE{5.8}{0}
    \HGATE{5.8}{1.6}
    \CZGATE{8}{1.6}{0}
\end{qcircuit}
\\
\begin{qcircuit}[scale=0.3]
    \blackgrid{0}{2.8}{0}
    \blackgrid{0}{2.8}{1.6}
    \BFOUR{1.4}{1.6}{0}
\end{qcircuit}
\raisebox{1em}{$=$}
\begin{qcircuit}[scale=0.3]
    \blackgrid{0}{2.8}{0}
    \blackgrid{0}{2.8}{1.6}
    \BEIG{1.4}{1.6}{0}
\end{qcircuit}
\raisebox{1em}{$=$}
\begin{qcircuit}[scale=0.3]
    \blackgrid{0}{13.8}{0,1.6}
    \HGATE{1.4}{0}
    \CZGATE{3.6}{1.6}{0}
    \HGATE{5.8}{0}
    \CZGATE{8}{1.6}{0}
    \HGATE{10.2}{0}
    \HGATE{10.2}{1.6}
    \CZGATE{12.4}{1.6}{0}
\end{qcircuit}
\\
\end{array}
\]
\end{definition}

\begin{definition}
\label{def:cgates}
The \emph{derived generators of kind $C$} are defined below.
\[
\begin{array}{ll}
\begin{qcircuit}[scale=0.3]
    \blackgrid{0}{2.8}{0}
    \CONE{1.4}{0}
\end{qcircuit} \ \ 
\raisebox{0.3em}{$=$}
\raisebox{0.65em}{
\begin{qcircuit}[scale=0.3]
    \blackgrid{0}{2.8}{0}
\end{qcircuit}}
& \ \ \ \ \ \
\begin{qcircuit}[scale=0.3]
    \blackgrid{0}{2.8}{0}
    \CTWO{1.4}{0}
\end{qcircuit} \ \
\raisebox{0.35em}{$=$}
\raisebox{0.25em}{
\begin{qcircuit}[scale=0.3]
    \blackgrid{0}{2.8}{0}
    \XGATE{1.4}{0}
\end{qcircuit}}
\end{array}
\]
\end{definition}

\begin{definition}
\label{def:dgates}
The \emph{derived generators of kind $D$} are defined below.
\[
\begin{array}{ll}
\begin{qcircuit}[scale=0.3]
    \blackgrid{0}{2.8}{0,1.6}
    \DONE{1.4}{1.6}{0}
\end{qcircuit}
\raisebox{1em}{$=$}
\begin{qcircuit}[scale=0.3]
    \blackgrid{0}{13.8}{0,1.6}
    \CZGATE{1.4}{1.6}{0}
    \HGATE{3.6}{0}
    \HGATE{3.6}{1.6}
    \CZGATE{5.8}{1.6}{0}
    \HGATE{8}{0}
    \HGATE{8}{1.6}
    \CZGATE{10.2}{1.6}{0}
    \HGATE{12.4}{0}
\end{qcircuit}
\\
\begin{qcircuit}[scale=0.3]
    \blackgrid{0}{2.8}{0,1.6}
    \DTWO{1.4}{1.6}{0}
\end{qcircuit}
\raisebox{1em}{$=$}
\begin{qcircuit}[scale=0.3]
    \blackgrid{0}{11.6}{0,1.6}
    \HGATE{1.4}{1.6}
    \CZGATE{3.6}{1.6}{0}
    \HGATE{5.8}{0}
    \HGATE{5.8}{1.6}
    \CZGATE{8}{1.6}{0}
    \HGATE{10.2}{0}
\end{qcircuit}
\\
\begin{qcircuit}[scale=0.3]
    \blackgrid{0}{2.8}{0,1.6}
    \DTHREE{1.4}{1.6}{0}
\end{qcircuit}
\raisebox{1em}{$=$}
\begin{qcircuit}[scale=0.3]
    \blackgrid{0}{11.6}{0,1.6}
    \HGATE{1.4}{0}
    \HGATE{1.4}{1.6}
    \CZGATE{3.6}{1.6}{0}
    \HGATE{5.8}{0}
    \HGATE{5.8}{1.6}
    \CZGATE{8}{1.6}{0}
    \HGATE{10.2}{0}
\end{qcircuit}
\\
\begin{qcircuit}[scale=0.3]
    \blackgrid{0}{2.8}{0,1.6}
    \DFOUR{1.4}{1.6}{0}
\end{qcircuit}
\raisebox{1em}{$=$}
\begin{qcircuit}[scale=0.3]
    \blackgrid{0}{13.8}{0,1.6}
    \HGATE{1.4}{1.6}
    \CZGATE{3.6}{1.6}{0}
    \HGATE{5.8}{0}
    \HGATE{5.8}{1.6}
    \CZGATE{8}{1.6}{0}
    \HGATE{10.2}{0}
    \CZGATE{12.4}{1.6}{0}
\end{qcircuit}
\end{array}
\]
\end{definition}

\begin{definition}
\label{def:egates}
The \emph{derived generators of kind $E$} are defined below.
\[
\begin{array}{ll}
\begin{qcircuit}[scale=0.3]
    \blackgrid{0}{2.8}{0}
    \EONE{1.4}{0}
\end{qcircuit}\ \ 
\raisebox{0.3em}{$=$}
\raisebox{0.64em}{
\begin{qcircuit}[scale=0.3]
    \blackgrid{0}{2.8}{0}
\end{qcircuit}}
& \ \ \ \ \ \
\begin{qcircuit}[scale=0.3]
    \blackgrid{0}{2.8}{0}
    \ETWO{1.4}{0}
\end{qcircuit}\ \ 
\raisebox{0.3em}{$=$}
\raisebox{0.44em}{
\begin{qcircuit}[scale=0.3]
    \blackgrid{0}{2.8}{0}
    \ZGATE{1.4}{0}
\end{qcircuit}}
\end{array}
\]
\end{definition}

\begin{definition}
The \emph{typed derived generators of kind $A$, $B$, $C$, $D$, and
$E$} are defined below.
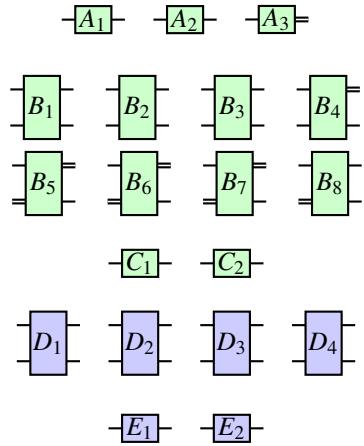
\[
\begin{array}{lll}
\begin{qcircuit}[scale=0.3]
    \blackgrid{0}{1.4}{0}
    \greengrid{1.4}{2.8}{0}
    \AONE{1.4}{0}
\end{qcircuit} 
& 
\begin{qcircuit}[scale=0.3]
    \blackgrid{0}{1.4}{0}
    \greengrid{1.4}{2.8}{0}
    \ATWO{1.4}{0}
\end{qcircuit}
&
\begin{qcircuit}[scale=0.3]
    \blackgrid{0}{1.4}{0}
    \redgrid{1.4}{2.8}{0}
    \ATHREE{1.4}{0}
\end{qcircuit}
\end{array}
\]
\[
\begin{array}{l}
\begin{array}{llll}
\begin{qcircuit}[scale=0.3]
    \greengrid{0}{1.4}{0}
    \blackgrid{1.4}{2.8}{0}
    \blackgrid{0}{1.4}{1.6}
    \greengrid{1.4}{2.8}{1.6}
    \BONE{1.4}{1.6}{0}
\end{qcircuit}
&
\begin{qcircuit}[scale=0.3]
    \greengrid{0}{1.4}{0}
    \blackgrid{1.4}{2.8}{0}
    \blackgrid{0}{1.4}{1.6}
    \greengrid{1.4}{2.8}{1.6}
    \BTWO{1.4}{1.6}{0}
\end{qcircuit}
&
\begin{qcircuit}[scale=0.3]
    \greengrid{0}{1.4}{0}
    \blackgrid{1.4}{2.8}{0}
    \blackgrid{0}{1.4}{1.6}
    \greengrid{1.4}{2.8}{1.6}
    \BTHREE{1.4}{1.6}{0}
\end{qcircuit}
&
\begin{qcircuit}[scale=0.3]
    \greengrid{0}{1.4}{0}
    \blackgrid{1.4}{2.8}{0}
    \blackgrid{0}{1.4}{1.6}
    \redgrid{1.4}{2.8}{1.6}
    \BFOUR{1.4}{1.6}{0}
\end{qcircuit} \\
\begin{qcircuit}[scale=0.3]
    \redgrid{0}{1.4}{0}
    \blackgrid{1.4}{2.8}{0}
    \blackgrid{0}{1.4}{1.6}
    \redgrid{1.4}{2.8}{1.6}
    \BFIVE{1.4}{1.6}{0}
\end{qcircuit}
&
\begin{qcircuit}[scale=0.3]
    \redgrid{0}{1.4}{0}
    \blackgrid{1.4}{2.8}{0}
    \blackgrid{0}{1.4}{1.6}
    \redgrid{1.4}{2.8}{1.6}
    \BSIX{1.4}{1.6}{0}
\end{qcircuit}
&
\begin{qcircuit}[scale=0.3]
    \redgrid{0}{1.4}{0}
    \blackgrid{1.4}{2.8}{0}
    \blackgrid{0}{1.4}{1.6}
    \redgrid{1.4}{2.8}{1.6}
    \BSEV{1.4}{1.6}{0}
\end{qcircuit}
&
\begin{qcircuit}[scale=0.3]
    \redgrid{0}{1.4}{0}
    \blackgrid{1.4}{2.8}{0}
    \blackgrid{0}{1.4}{1.6}
    \greengrid{1.4}{2.8}{1.6}
    \BEIG{1.4}{1.6}{0}
\end{qcircuit}
\end{array}
\end{array}
\]
\[
\begin{array}{ll}
\begin{qcircuit}[scale=0.3]
    \greengrid{0}{1.4}{0}
    \blackgrid{1.4}{2.8}{0}
    \CONE{1.4}{0}
\end{qcircuit}
&
\begin{qcircuit}[scale=0.3]
    \greengrid{0}{1.4}{0}
    \blackgrid{1.4}{2.8}{0}
    \CTWO{1.4}{0}
\end{qcircuit}
\end{array}
\]
\[
\begin{array}{llll}
\begin{qcircuit}[scale=0.3]
    \blackgrid{0}{2.8}{0,1.6}
    \DONE{1.4}{1.6}{0}
\end{qcircuit}
&
\begin{qcircuit}[scale=0.3]
    \blackgrid{0}{2.8}{0,1.6}
    \DTWO{1.4}{1.6}{0}
\end{qcircuit}
&
\begin{qcircuit}[scale=0.3]
    \blackgrid{0}{2.8}{0,1.6}
    \DTHREE{1.4}{1.6}{0}
\end{qcircuit}
&
\begin{qcircuit}[scale=0.3]
    \blackgrid{0}{2.8}{0,1.6}
    \DFOUR{1.4}{1.6}{0}
\end{qcircuit}
\end{array}
\]

\[
\begin{array}{ll}
\begin{qcircuit}[scale=0.3]
    \blackgrid{0}{2.8}{0}
    \EONE{1.4}{0}
\end{qcircuit}
&
\begin{qcircuit}[scale=0.3]
    \blackgrid{0}{2.8}{0}
    \ETWO{1.4}{0}
\end{qcircuit}
\end{array}
\]
\end{definition}

\begin{proposition}
\label{genActions}
The following annotated circuits record the action of the derived
generators of kind $A$, $B$, $C$, $D$, and $E$ on certain Pauli
operators.
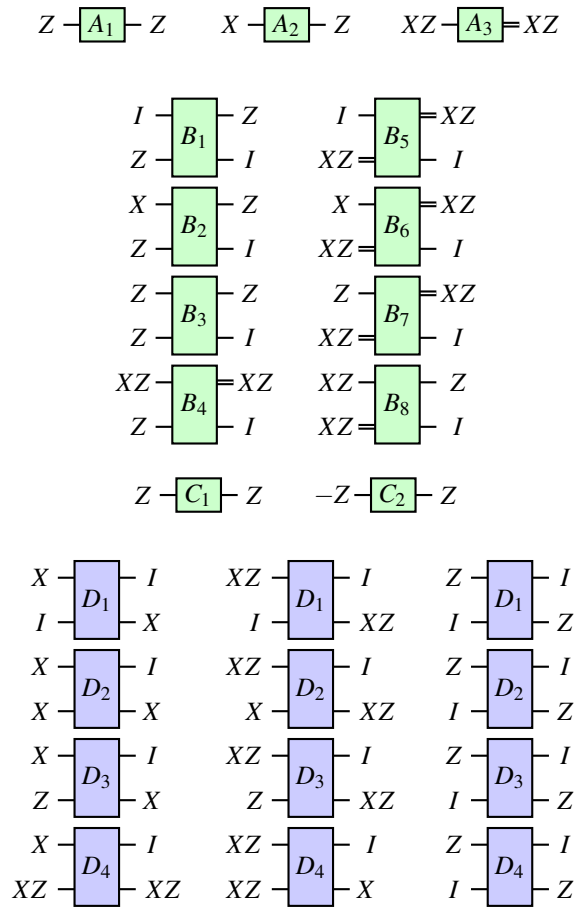
\[
\begin{array}{lll}
\begin{qcircuit}[scale=0.37]
    \LABEL{-0.6}{0}{$Z$}
    \LABEL{3.4}{0}{$Z$}
    \blackgrid{0}{1.4}{0}
    \greengrid{1.4}{2.8}{0}
    \AONE{1.4}{0}
\end{qcircuit}
&
\begin{qcircuit}[scale=0.37]
    \LABEL{-0.6}{0}{$X$}
    \LABEL{3.4}{0}{$Z$}
    \blackgrid{0}{1.4}{0}
    \greengrid{1.4}{2.8}{0}
    \ATWO{1.4}{0}
\end{qcircuit}
&
\begin{qcircuit}[scale=0.37]
    \LABEL{-0.8}{0}{$XZ$}
    \LABEL{3.6}{0}{$XZ$}
    \blackgrid{0}{1.4}{0}
    \redgrid{1.4}{2.8}{0}
    \ATHREE{1.4}{0}
\end{qcircuit}
\end{array}
\]

\[
\begin{array}{ll}
\begin{qcircuit}[scale=0.37]
    \LABEL{-0.6}{0}{$Z$}
    \LABEL{-0.8}{1.6}{$XZ$}
    \LABEL{3.4}{0}{$I$}
    \LABEL{3.6}{1.6}{$XZ$}
    \greengrid{0}{1.4}{0}
    \blackgrid{1.4}{2.8}{0}
    \blackgrid{0}{1.4}{1.6}
    \redgrid{1.4}{2.8}{1.6}
    \BFOUR{1.4}{1.6}{0}

    \LABEL{-0.6}{3.2}{$Z$}
    \LABEL{-0.6}{4.8}{$Z$}
    \LABEL{3.4}{3.2}{$I$}
    \LABEL{3.4}{4.8}{$Z$}
    \greengrid{0}{1.4}{3.2}
    \blackgrid{1.4}{2.8}{3.2}
    \blackgrid{0}{1.4}{4.8}
    \greengrid{1.4}{2.8}{4.8}
    \BTHREE{1.4}{4.8}{3.2}

    \LABEL{-0.6}{6.4}{$Z$}
    \LABEL{-0.6}{8}{$X$}
    \LABEL{3.4}{6.4}{$I$}
    \LABEL{3.4}{8}{$Z$}
    \greengrid{0}{1.4}{6.4}
    \blackgrid{1.4}{2.8}{6.4}
    \blackgrid{0}{1.4}{8}
    \greengrid{1.4}{2.8}{8}
    \BTWO{1.4}{8}{6.4}

    \LABEL{-0.6}{9.6}{$Z$}
    \LABEL{-0.6}{11.2}{$I$}
    \LABEL{3.4}{9.6}{$I$}
    \LABEL{3.4}{11.2}{$Z$}
    \greengrid{0}{1.4}{9.6}
    \blackgrid{1.4}{2.8}{9.6}
    \blackgrid{0}{1.4}{11.2}
    \greengrid{1.4}{2.8}{11.2}
    \BONE{1.4}{11.2}{9.6}
\end{qcircuit}
&
\begin{qcircuit}[scale=0.37]
    \LABEL{-0.8}{0}{$XZ$}
    \LABEL{-0.8}{1.6}{$XZ$}
    \LABEL{3.6}{0}{$I$}
    \LABEL{3.6}{1.6}{$Z$}
    \redgrid{0}{1.4}{0}
    \blackgrid{1.4}{2.8}{0}
    \blackgrid{0}{1.4}{1.6}
    \greengrid{1.4}{2.8}{1.6}
    \BEIG{1.4}{1.6}{0}

    \LABEL{-0.8}{3.2}{$XZ$}
    \LABEL{-0.6}{4.8}{$Z$}
    \LABEL{3.6}{3.2}{$I$}
    \LABEL{3.6}{4.8}{$XZ$}
    \redgrid{0}{1.4}{3.2}
    \blackgrid{1.4}{2.8}{3.2}
    \blackgrid{0}{1.4}{4.8}
    \redgrid{1.4}{2.8}{4.8}
    \BSEV{1.4}{4.8}{3.2}

    \LABEL{-0.8}{6.4}{$XZ$}
    \LABEL{-0.6}{8}{$X$}
    \LABEL{3.6}{6.4}{$I$}
    \LABEL{3.6}{8}{$XZ$}
    \redgrid{0}{1.4}{6.4}
    \blackgrid{1.4}{2.8}{6.4}
    \blackgrid{0}{1.4}{8}
    \redgrid{1.4}{2.8}{8}
    \BSIX{1.4}{8}{6.4}

    \LABEL{-0.8}{9.6}{$XZ$}
    \LABEL{-0.6}{11.2}{$I$}
    \LABEL{3.6}{9.6}{$I$}
    \LABEL{3.6}{11.2}{$XZ$}
    \redgrid{0}{1.4}{9.6}
    \blackgrid{1.4}{2.8}{9.6}
    \blackgrid{0}{1.4}{11.2}
    \redgrid{1.4}{2.8}{11.2}
    \BFIVE{1.4}{11.2}{9.6}
\end{qcircuit}
\end{array}
\]

\[
\begin{array}{ll}
\begin{qcircuit}[scale=0.37]
    \LABEL{-0.6}{0}{$Z$}
    \LABEL{3.4}{0}{$Z$}
    \greengrid{0}{1.4}{0}
    \blackgrid{1.4}{2.8}{0}
    \CONE{1.4}{0}
\end{qcircuit}
&
\begin{qcircuit}[scale=0.37]
    \LABEL{-0.8}{0}{$-Z$}
    \LABEL{3.4}{0}{$Z$}
    \greengrid{0}{1.4}{0}
    \blackgrid{1.4}{2.8}{0}
    \CTWO{1.4}{0}
\end{qcircuit}
\end{array}
\]

\[
\begin{array}{llll}
\begin{qcircuit}[scale=0.37]
    \LABEL{-1}{0}{$XZ$}
    \LABEL{-0.6}{1.6}{$X$}
    \LABEL{3.8}{0}{$XZ$}
    \LABEL{3.4}{1.6}{$I$}
    \blackgrid{0}{2.8}{0,1.6}
    \DFOUR{1.4}{1.6}{0}

    \LABEL{-0.6}{3.2}{$Z$}
    \LABEL{-0.6}{4.8}{$X$}
    \LABEL{3.4}{3.2}{$X$}
    \LABEL{3.4}{4.8}{$I$}
    \blackgrid{0}{2.8}{3.2,4.8}
    \DTHREE{1.4}{4.8}{3.2}

    \LABEL{-0.6}{6.4}{$X$}
    \LABEL{-0.6}{8}{$X$}
    \LABEL{3.4}{6.4}{$X$}
    \LABEL{3.4}{8}{$I$}
    \blackgrid{0}{2.8}{6.4,8}
    \DTWO{1.4}{8}{6.4}

    \LABEL{-0.6}{9.6}{$I$}
    \LABEL{-0.6}{11.2}{$X$}
    \LABEL{3.4}{9.6}{$X$}
    \LABEL{3.4}{11.2}{$I$}
    \blackgrid{0}{2.8}{9.6,11.2}
    \DONE{1.4}{11.2}{9.6}
\end{qcircuit}
&
\begin{qcircuit}[scale=0.37]
    \LABEL{-1}{0}{$XZ$}
    \LABEL{-1}{1.6}{$XZ$}
    \LABEL{3.4}{0}{$X$}
    \LABEL{3.6}{1.6}{$I$}
    \blackgrid{0}{2.8}{0,1.6}
    \DFOUR{1.4}{1.6}{0}

    \LABEL{-0.6}{3.2}{$Z$}
    \LABEL{-1}{4.8}{$XZ$}
    \LABEL{3.8}{3.2}{$XZ$}
    \LABEL{3.4}{4.8}{$I$}
    \blackgrid{0}{2.8}{3.2,4.8}
    \DTHREE{1.4}{4.8}{3.2}

    \LABEL{-0.6}{6.4}{$X$}
    \LABEL{-1}{8}{$XZ$}
    \LABEL{3.8}{6.4}{$XZ$}
    \LABEL{3.4}{8}{$I$}
    \blackgrid{0}{2.8}{6.4,8}
    \DTWO{1.4}{8}{6.4}

    \LABEL{-0.6}{9.6}{$I$}
    \LABEL{-1}{11.2}{$XZ$}
    \LABEL{3.8}{9.6}{$XZ$}
    \LABEL{3.4}{11.2}{$I$}
    \blackgrid{0}{2.8}{9.6,11.2}
    \DONE{1.4}{11.2}{9.6}
\end{qcircuit}
&
\begin{qcircuit}[scale=0.37]
    \LABEL{-0.6}{0}{$I$}
    \LABEL{-0.6}{1.6}{$Z$}
    \LABEL{3.4}{0}{$Z$}
    \LABEL{3.4}{1.6}{$I$}
    \blackgrid{0}{2.8}{0,1.6}
    \DFOUR{1.4}{1.6}{0}

    \LABEL{-0.6}{3.2}{$I$}
    \LABEL{-0.6}{4.8}{$Z$}
    \LABEL{3.4}{3.2}{$Z$}
    \LABEL{3.4}{4.8}{$I$}
    \blackgrid{0}{2.8}{3.2,4.8}
    \DTHREE{1.4}{4.8}{3.2}

    \LABEL{-0.6}{6.4}{$I$}
    \LABEL{-0.6}{8}{$Z$}
    \LABEL{3.4}{6.4}{$Z$}
    \LABEL{3.4}{8}{$I$}
    \blackgrid{0}{2.8}{6.4,8}
    \DTWO{1.4}{8}{6.4}

    \LABEL{-0.6}{9.6}{$I$}
    \LABEL{-0.6}{11.2}{$Z$}
    \LABEL{3.4}{9.6}{$Z$}
    \LABEL{3.4}{11.2}{$I$}
    \blackgrid{0}{2.8}{9.6,11.2}
    \DONE{1.4}{11.2}{9.6}
\end{qcircuit}
\end{array}
\]

\[
\begin{array}{llll}
\begin{qcircuit}[scale=0.37]
    \LABEL{-0.6}{0}{$X$}
    \LABEL{3.4}{0}{$X$}
    \blackgrid{0}{2.8}{0}
    \EONE{1.4}{0}
\end{qcircuit}
&
\begin{qcircuit}[scale=0.37]
    \LABEL{-0.6}{0}{$Z$}
    \LABEL{3.4}{0}{$Z$}
    \blackgrid{0}{2.8}{0}
    \EONE{1.4}{0}
\end{qcircuit}
&
\begin{qcircuit}[scale=0.37]
    \LABEL{-0.8}{0}{$-X$}
    \LABEL{3.4}{0}{$X$}
    \blackgrid{0}{2.8}{0}
    \ETWO{1.4}{0}
\end{qcircuit}
&
\begin{qcircuit}[scale=0.37]
    \LABEL{-0.8}{0}{$Z$}
    \LABEL{3.4}{0}{$Z$}
    \blackgrid{0}{2.8}{0}
    \ETWO{1.4}{0}
\end{qcircuit}
\end{array}
\]
In each case, the specified gate is the unique derived generator of
its kind and type that performs the specified action.
\end{proposition}

\begin{proof}
By computation.
\end{proof}

\subsection{Normal forms}

We now describe normal forms.

\begin{definition}
\label{def:znormal}
A typed $n$-qubit circuit is a \emph{$Z$-circuit} if it is of the form
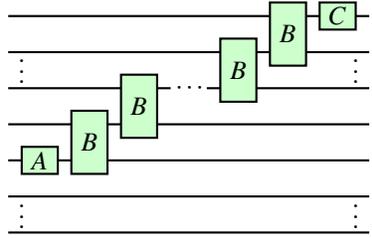
\[
\begin{qcircuit}[scale=0.3]
    \Period{16.4}{-3.2}
    \LABEL{.6}{4.3}{$\vdots$}
    \LABEL{15.4}{4.3}{$\vdots$}
    \LABEL{.6}{-2.13}{$\vdots$}
    \LABEL{15.4}{-2.13}{$\vdots$}
    \LABEL{8.1}{3.19}{$\cdots$}
    \blackgrid{0}{16}{-1.6}
    \blackgrid{0}{16}{-3.2}
    \blackgrid{0}{1.4}{0}
    \blackgrid{1.4}{3.6}{0}
    \blackgrid{3.6}{16}{0}
    
    \blackgrid{0}{3.6}{1.6}
    \blackgrid{3.6}{5.8}{1.6}
    \blackgrid{5.8}{16}{1.6}
    
    \blackgrid{0}{5.8}{3.2}
    \blackgrid{5.8}{7.2}{3.2}
    \blackgrid{8.8}{10.2}{3.2}
    \blackgrid{10.2}{16}{3.2}
    
    \blackgrid{0}{10.2}{4.8}
    \blackgrid{10.2}{12.4}{4.8}
    \blackgrid{12.4}{16}{4.8}
    
    \blackgrid{0}{12.4}{6.4}
    \greengrid{12.4}{14.6}{6.4}
    \blackgrid{14.6}{16}{6.4}
    
    \GENERICORANGE{1.4}{0}{0}{$A$}
    \GENERICGREEN{3.6}{1.6}{0}{$B$}
    \GENERICGREEN{5.8}{3.2}{1.6}{$B$}
    \GENERICGREEN{10.2}{4.8}{3.2}{$B$}
    \GENERICGREEN{12.4}{6.4}{4.8}{$B$}
    \GENERICYELLOW{14.6}{6.4}{6.4}{$C$}
\end{qcircuit}
\]
\end{definition}

In accordance with the convention introduced in \cref{sec:circuits},
the circuit in \cref{def:znormal} denotes a family of well-formed
typed circuits. The types of the derived generators of kind $B$ then
imply, for example, that if the first $B$ gate is $B_4$, then the
second $B$ gate can only be $B_5$, $B_6$, $B_7$, or $B_8$.

\begin{definition}
\label{def:xnormal}
A typed $n$-qubit circuit is an \emph{$X$-circuit} if it is of the
form
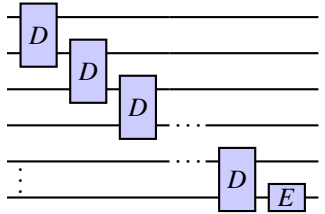
\[
\begin{qcircuit}[scale=0.3]
    \Period{14.2}{-1.6}
    \LABEL{8.1}{1.56}{$\cdots$}
    \LABEL{8.1}{-0.04}{$\cdots$}
    \LABEL{0.6}{-0.5}{$\vdots$}
    \blackgrid{0}{7.2}{0,1.6,3.2,4.8,6.4}
    \blackgrid{7.2}{13.8}{3.2,4.8,6.4}
    \blackgrid{8.8}{13.8}{1.6}
    \blackgrid{8.8}{13.8}{0}
    \blackgrid{0}{13.8}{-1.6}
    \GENERICBLUE{1.4}{6.4}{4.8}{$D$}
    \GENERICBLUE{3.6}{4.8}{3.2}{$D$}
    \GENERICBLUE{5.8}{3.2}{1.6}{$D$}
    \GENERICBLUE{10.2}{0}{-1.6}{$D$}
    \GENERICPINK{12.4}{-1.6}{-1.6}{$E$}
\end{qcircuit}
\]
\end{definition}

\begin{definition}
\label{def:normal}
A typed $n$-qubit circuit is \emph{normal} if it is of the form
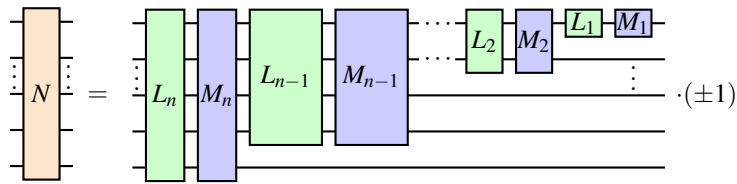
\[
\begin{array}{ll}
\raisebox{0.05em}{
\begin{qcircuit}[scale=0.3]
    \LABEL{3.9}{3.1}{$=$}
    \LABEL{2.6}{4.3}{$\vdots$}
    \LABEL{0.2}{4.3}{$\vdots$}
    \blackgrid{0}{2.8}{0,1.6,3.2,4.8,6.4}
    \colbiggate{orange!20}{$N$}{1.4,7.6}{1.4,-1.2}
\end{qcircuit}}
\begin{qcircuit}[scale=0.3]
    \LABEL{13.6}{6.36}{$\cdots$}
    \LABEL{13.6}{4.8}{$\cdots$}
    \LABEL{0.2}{4.3}{$\vdots$}
    \LABEL{22.2}{4.3}{$\vdots$}
    \LABEL{25.45}{3.2}{$\cdot (\pm 1)$}
    \blackgrid{0}{12.7}{0,1.6,3.2,4.8,6.4}
    \blackgrid{12.7}{23.6}{0,1.6,3.2}
    \blackgrid{14.3}{23.6}{6.4}
    \blackgrid{14.3}{23.6}{4.8}
    \colbiggate{green!20}{$L_{n}$}{1.4,0}{1.5,6.4}
    \colbiggate{blue!20}{$M_{n}$}{3.7,0}{3.8,6.4}
    \colbiggate{green!20}{$L_{n-1}$}{6,7.6}{7.6,.4}    
    \colbiggate{blue!20}{$M_{n-1}$}{9.8,7.6}{11.4,.4}        
    \colbiggate{green!20}{$L_{2}$}{15.6,4.8}{15.6,6.4}
    \colbiggate{blue!20}{$M_{2}$}{17.8,4.8}{17.8,6.4}
    \colbiggate{green!20}{$L_{1}$}{20,6.4}{20,6.4}
    \colbiggate{blue!20}{$M_{1}$}{22.2,6.4}{22.2,6.4}
\end{qcircuit}
\end{array}
\]
where, for $1 \leq i \leq n$, $L_{i}$ is a $Z$-circuit, and $M_{i}$ is
an $X$-circuit.
\end{definition}

The propositions below establish that every automorphism of
$\pauli{n}$ is represented by a unique normal Clifford circuit. Given
an automorphism $\phi$ of $\pauli{n}$ one can construct a $Z$-circuit
$L$ and an $X$-circuit $M$ such that $(ML)^{-1}$ acts as $\phi^{-1}$
on $I \otimes \cdots \otimes I \otimes Z$ and $ I \otimes \cdots
\otimes I \otimes X$. To obtain a normal form for $\phi$ it then
suffices to first construct $L$ and $M$ and then to recursively
proceed with the automorphism $\phi'=\phi(ML)^{-1}$. The proofs can be found in \cref{app:nfproofs}.

\begin{proposition}
  \label{zexistunique}
  Let P be a $n$-qubit Pauli operator, with $P = P_{1} \otimes P_{2}
  \otimes \cdots \otimes P_{n}$, $P^2 = I$, and $P \neq \pm I$. Then
  there exists a unique $Z$-circuit $L$ such that $L \bullet P = Z
  \otimes I \otimes \cdots \otimes I$.
\end{proposition}

\begin{proposition}
  \label{xexistunique}
  Let $Q$ be an $n$-qubit Pauli operator with $Q= Q_1\otimes
  Q_2\otimes \cdots \otimes Q_n$, $Q^{2}=I$, $Q \neq \pm I$, and $Q$
  anticommutes with $Z \otimes I \otimes \cdots \otimes I$. Then there
  exists a unique $X$-circuit $M$ such that $M \bullet Q = I \otimes
  \cdots \otimes I \otimes X$.
\end{proposition}

\begin{proposition}
\label{xonz}
Every $X$-circuit $M$ satisfies $M \bullet (Z \otimes I \otimes \cdots
\otimes I) = I \otimes \cdots \otimes I \otimes Z$.
\end{proposition}

\begin{proposition}
  \label{prop:existlm}
  Let $P$ and $Q$ be Pauli operators such that $P^2=Q^2=I,$ $P,Q \neq
  \pm I$, and $P$ and $Q$ anticommute. Then there exists a unique pair
  of a $Z$-circuit $L$ and a $X$-circuit $M$ such that $ML \bullet P =
  I \otimes \cdots \otimes I \otimes Z$ and $ML \bullet Q = I \otimes
  \cdots \otimes I \otimes X$.
\end{proposition}

\begin{proposition}
  \label{autosimulate}
  Let $\phi : \pauli{n} \rightarrow \pauli{n}$ be an automorphism of
  the Pauli group. Then there exists a normal circuit $C$ such that
  for all $P$, $C \bullet P = \phi(P)$. Moreover, the normal form $C$
  is unique up to a scalar $\pm 1$.
\end{proposition}

By the existence part of \cref{autosimulate}, every automorphism of
the Pauli group can be represented as a circuit over $-1$, $H$, $Z$,
and $CZ $. Thus, all of these automorphisms are
stabilizers. Conversely, as remarked in \cref{sec:groups}, every
stabilizer is an automorphism of the Pauli group. Hence,
\cref{autosimulate} indeed establishes that every stabilizer admits a
unique normal form. Note that this also proves that stabilizers are
generated by $-1$, $Z$, $H$, and $CZ $.

By \cref{autosimulate}, there is a bijection between stabilizer
operators and normal forms. We can therefore count the number of
$n$-qubit normal forms to compute the cardinality of $\clifford{n}$.

\begin{corollary}
  \label{cor:numberofnfs}
  There are exactly $2 \cdot \prod_{i=1}^{n} (4^{i} + 2^{i} -2)(2
  \cdot 4^{i-1})$ real stabilizer operators on $n$ qubits.
\end{corollary}

\begin{proof}
  See \cref{app:cliffsize}
\end{proof}

% --------------------------------------------------------------------
\section{Relations for Real Stabilizer Circuits}
\label{sec:rels}

We now introduce relations for real Clifford circuits and describe an
algorithm for converting any $n$-qubit Clifford circuit to its normal
form, using finitely many applications of the relations. To normalize
circuits, it is sufficient to have relations to
\begin{enumerate}
\item rewrite the empty circuit into the normal form for the identity
  and
\item rewrite a circuit consisting of a single gate appearing on the
  left of a normal form into a normal form.
\end{enumerate}
Indeed, one can then start with an arbitrary circuit, append the
normal form for the identity to the right of it, and iteratively merge
the gates of the initial circuit into the normal form on its right.

\begin{definition}
  \label{def:rels}
  The \emph{typed relations} for real stabilizers are given in
  \cref{app:rels}.
\end{definition}

The typed relations describe all situations in which one of $H$, $Z$,
$CZ$, $X$, or $CXZ$ appears to the left of a normal form. Because
these gates act on no more than two qubits, there are only finitely
many cases to consider. The difficulty arises because the right-hand
side of a relation may contain multiple gates. As a result, we are led
to consider cases where a circuit appears on the left-hand side of a
rule. This process increases the number of cases to consider and
could, in principle, fail to terminate. However, a careful analysis
shows that this is not the case. In total, 139 relations are contained
in \cref{app:rels}.

\subsection{Normalization}

We start by labelling normal circuits. This labelling is convenient to
refer to specific parts of a circuit and will be useful to describe
our rewrite system.

\begin{definition}
Consider an $n$-qubit normal circuit
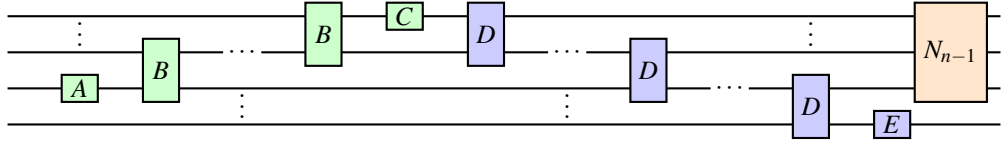
\[
\begin{qcircuit}[scale=0.3]
    \LABEL{1.4}{2.75}{$\vdots$}
    \LABEL{33.8}{2.75}{$\vdots$}
    \LABEL{8.6}{-0.45}{$\vdots$}
    \LABEL{23}{-0.45}{$\vdots$}        
    \blackgrid{-1.8}{42.2}{-1.6,0,1.6,3.2}  
    \LABEL{8.6}{1.6}{$\hspace{0.1em}\cdots$}
    \LABEL{23}{1.6}{$\hspace{0.1em}\cdots$}
    \LABEL{30.2}{0}{$\hspace{0.1em}\cdots$}        
    \GENERICORANGE{1.4}{0}{0}{$A$}
    \GENERICGREEN{5}{1.6}{0}{$B$}
    \GENERICGREEN{12.2}{3.2}{1.6}{$B$}
    \GENERICYELLOW{15.8}{3.2}{3.2}{$C$}
    \GENERICBLUE{19.4}{3.2}{1.6}{$D$}
    \GENERICBLUE{26.6}{1.6}{0}{$D$}
    \GENERICBLUE{33.8}{0}{-1.6}{$D$}
    \GENERICPINK{37.4}{-1.6}{-1.6}{$E$}
    \colbiggate{orange!20}{$N_{n-1}$}{39.2,0}{40.8,3.2}
\end{qcircuit}
\]
where $N_{n-1}$ is assumed to be a normal form on $(n-1)$ qubits. We
assign labels to specific wires in order to produce a \emph{labelled
normal form}
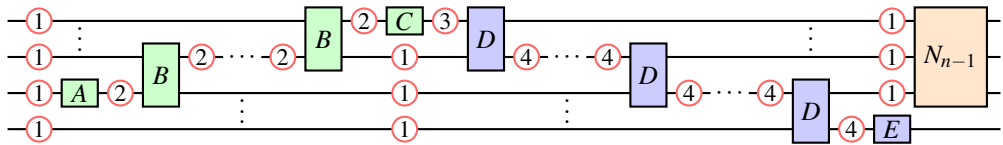
\[
\begin{qcircuit}[scale=0.3]
    \LABEL{1.4}{2.75}{$\vdots$}
    \LABEL{33.8}{2.75}{$\vdots$}
    \LABEL{8.6}{-0.45}{$\vdots$}
    \LABEL{23}{-0.45}{$\vdots$}        
    \blackgrid{-1.8}{42.2}{-1.6,0,1.6,3.2}  
    \wirelabel{-0.4}{-1.6}{1}
    \wirelabel{-0.4}{0}{1}
    \wirelabel{-0.4}{1.6}{1}
    \wirelabel{-0.4}{3.2}{1}
    \LABEL{8.6}{1.6}{$\hspace{0.1em}\cdots$}
    \LABEL{23}{1.6}{$\hspace{0.1em}\cdots$}
    \LABEL{30.2}{0}{$\hspace{0.1em}\cdots$}            
    \GENERICORANGE{1.4}{0}{0}{$A$}
    \wirelabel{3.2}{0}{2}
    \GENERICGREEN{5}{1.6}{0}{$B$}
    \wirelabel{6.8}{1.6}{2}
    \wirelabel{10.4}{1.6}{2}    
    \GENERICGREEN{12.2}{3.2}{1.6}{$B$}
    \wirelabel{14}{3.2}{2}
    \GENERICYELLOW{15.8}{3.2}{3.2}{$C$}
    \wirelabel{17.6}{3.2}{3}
    \wirelabel{15.8}{-1.6}{1}
    \wirelabel{15.8}{0}{1}
    \wirelabel{15.8}{1.6}{1}
    \GENERICBLUE{19.4}{3.2}{1.6}{$D$}
    \wirelabel{21.2}{1.6}{4}
    \wirelabel{24.8}{1.6}{4}    
    \GENERICBLUE{26.6}{1.6}{0}{$D$}
    \wirelabel{28.4}{0}{4}
    \wirelabel{32}{0}{4}    
    \GENERICBLUE{33.8}{0}{-1.6}{$D$}
    \wirelabel{35.6}{-1.6}{4}
    \GENERICPINK{37.4}{-1.6}{-1.6}{$E$}
    \wirelabel{37.4}{0}{1}
    \wirelabel{37.4}{1.6}{1}
    \wirelabel{37.4}{3.2}{1}    
    \colbiggate{orange!20}{$N_{n-1}$}{39.2,0}{40.8,3.2}    
\end{qcircuit}
\]
where $N_{(n-1)}$ is recursively labelled in the same
manner.
\end{definition}

\begin{definition}
  \label{def:dirty}
  \emph{Dirty normal forms} are obtained from normal forms by adding
  gates according to the following scheme.
  \begin{itemize}
  \item An $H$ gate can be placed on a wire labelled 1 or on a double
    wire labelled 2.
  \item A $Z$ gate can be placed on a wire labelled 1, 2, 3, or 4.
  \item An $X$ gate can be placed on a wire labelled 1 or 2.
  \item A $CZ $ gate can be placed on adjacent wires, provided that
    the bottom wire is labelled 1, and either the top wire is labelled
    1 or 3 or the top wire is a double wire labelled 2.
  \item A $CXZ $ gate can be placed on adjacent wires, provided that
    the bottom wire is labelled 1, and the top wire is a double wire
    labelled 2.
  \end{itemize}
  When discussing dirty normal forms, we call the $H$, $Z$, $X$, $CZ
  $, and $CXZ $ gates \emph{dirty}, while the gates of kind $A$, $B$,
  $C$, $D$, and $E$ are called \emph{clean}.
\end{definition}

Intuitively, dirty normal forms are circuits ``during the
normalization process''. We now explain how the relations can be used
to transform dirty normal forms into clean ones.

\begin{lemma}
\label{biglemma}
Any dirty normal form can be converted to its normal form by applying
the typed relations of \cref{def:rels} a finite number of times.
\end{lemma}

\begin{proof}
  By \cref{def:dirty}, every dirty gate occurs before a clean
  gate. Thus, if dirty gates remain in the circuit, a dirty gate must
  occur immediately before a clean one. The left-hand side of the
  typed relations of \cref{def:rels} contain all cases of a dirty gate
  occurring immediately before a clean gate. Hence, as long as dirty
  gates remain, one of the rules can be applied. Moreover, each rule
  takes a dirty normal form to a dirty normal form. We now show that
  this procedure terminates in a finite number of steps. To this end,
  we associate a sequence of nonnegative integers numbers to each
  dirty normal form. Suppose a dirty normal form has $t$ clean gates,
  indexed $1,\ldots,t$ left to right. Now define the sequence $s =
  (s_1,\ldots,s_t)$ where $s_i$ is the number of dirty gates that
  occur before the $i$-th clean gate. A left-to-right application of
  one of the typed relations decreases $s$ lexicographically. The
  length of this sequence might not remain constant through the
  normalization process but it is bounded by the maximum possible
  number of clean gates in a circuit. For a normal form on $n$-qubits,
  this bound is given by
  \[
  \sum_{i=1}^{n} 2i+1 = n^{2}+2n.
  \]
  Hence, this process terminates in a finite number of rewrites.
\end{proof}

\begin{proposition}
  \label{prop:normalization}
  Any Clifford circuit can be rewritten into its normal form using the
  typed relations of \cref{def:rels}.
\end{proposition}

\begin{proof}
The normal form of the identity operator on $n$ qubits is 
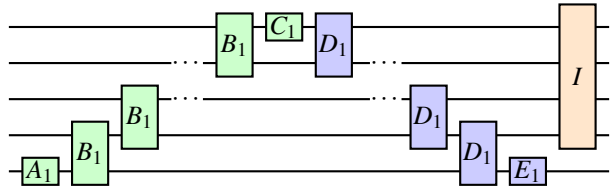
\[
\begin{qcircuit}[scale=0.3]
    \LABEL{7.9}{3.2}{$\cdots$}
    \LABEL{7.9}{4.8}{$\cdots$}
    \LABEL{16.7}{3.2}{$\cdots$}
    \LABEL{16.7}{4.8}{$\cdots$}
    
    \blackgrid{0}{1.4}{0}
    \greengrid{1.4}{2.8}{0}
    \blackgrid{2.8}{26.6}{0}
    
    \blackgrid{0}{3.6}{1.6}
    \greengrid{3.6}{5.8}{1.6}
    \blackgrid{5.8}{26.6}{1.6}
    
    \blackgrid{0}{5.8}{3.2}
    \greengrid{5.8}{7.2}{3.2}
    \blackgrid{8.6}{16}{3.2}
    \blackgrid{17.4}{26.6}{3.2}
    
    \blackgrid{0}{7.2}{4.8}
    \greengrid{8.6}{10}{4.8}
    \blackgrid{10}{16}{4.8}
    \blackgrid{17.4}{26.6}{4.8}
    
    \blackgrid{0}{10}{6.4}
    \greengrid{10}{12.2}{6.4}
    \blackgrid{12.2}{26.6}{6.4}
    
    \AONE{1.4}{0}
    \BONE{3.6}{1.6}{0}
    \BONE{5.8}{3.2}{1.6}
    \BONE{10}{6.4}{4.8}
    \CONE{12.2}{6.4}
    \DONE{14.4}{6.4}{4.8}
    \DONE{18.6}{3.2}{1.6}
    \DONE{20.8}{1.6}{0}
    \EONE{23}{0}
    
    \colbiggate{orange!20}{$I$}{25.2,8}{25.2,0.4}
    
\end{qcircuit}
\]
where $I$ denotes the normal form for the identity on $n-1$
qubits. Using the typed relations of \cref{def:rels}, we can rewrite
the empty circuit on $n$ wires into the normal form for the
identity. Now consider a Clifford circuit $C$. Expanding the wires on
the right of $C$ into the normal form for the identity, we obtain a
dirty normal form. We can then convert this dirty normal form into a
normal form using \cref{biglemma}, which completes the proof.
\end{proof}

\subsection{A Reduced Set of Relations}
\label{ssec:compactrels}

\cref{autosimulate,prop:normalization} jointly show that real
stabilizers are presented by the generators $(-1)$, $H$, $Z$, and $CZ
$ and the typed relations of \cref{def:rels} (where each derived
generator is replaced by its definition and types are forgotten). This
presentation is highly redundant and, in this final section, we
provide a reduced set of relations.

\begin{definition}
  \label{def:reducedrels}
  The \emph{reduced relations} for real stabilizers are given in
  \cref{fig:compactrels}.
\end{definition}

\begin{figure}
  \input{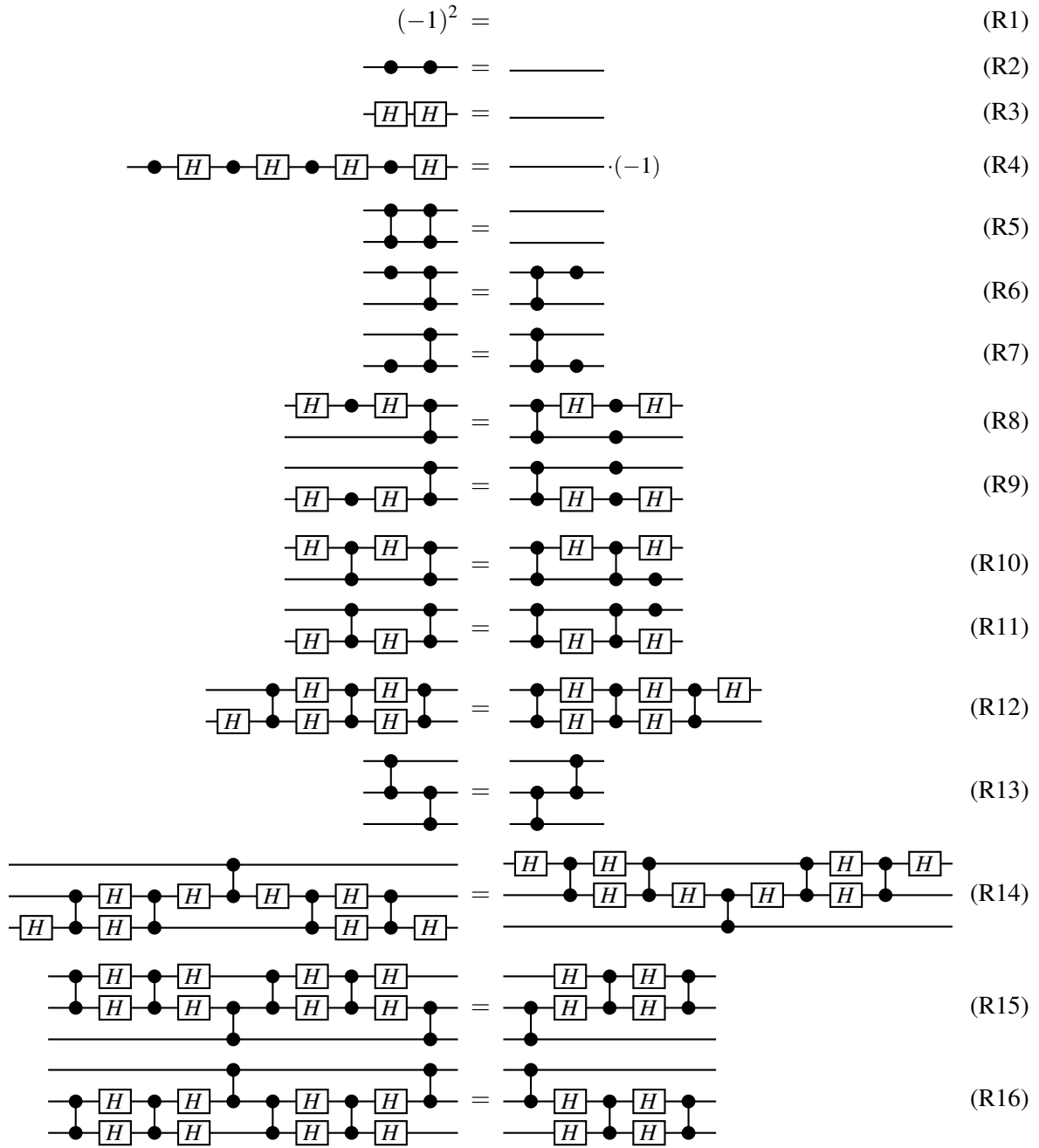}
  \caption{A set of reduced relations for real stabilizers.}
  \label{fig:compactrels}
\end{figure}

\begin{proposition}
  \label{prop:compactnormalization}
  Any Clifford circuit can be rewritten into its normal form using the
  reduced relations of \cref{def:reducedrels}.
\end{proposition}

\begin{proof}
It suffices to show that the reduced relations of
\cref{def:reducedrels} imply the typed relations of
\cref{def:rels}. The derivations can be found in the supplement to this paper
\cite{SupplementRealStab}.
\end{proof}

An alternative set of reduced relations is given in
\cref{app:relsalt}. This last collection of relations is stated in
terms of the $(-1)$, $H$, $Z$, $CZ $, $X$, and $CX$ gates and is
included because it makes for an arguably more intuitive presentation.

% --------------------------------------------------------------------
\section{Conclusion}
\label{sec:conc}

In this paper, we defined a normal form for real stabilizer circuits,
showed that every real stabilizer operator admits a unique normal
form, and introduced a set of relations that suffices to rewrite any
real stabilizer circuit into its normal form. This yields a
presentation by generators and relations of real stabilizer
operators. Our results add to the growing family of quantum operators
for which such presentations are known (see \cite{Amy16,Selinger15}
for unitary quantum circuits and, for example,
\cite{Backens2014TheZI,Comfort19,Vilmart2018AZW} in more general
contexts). Our approach in this work followed that of
\cite{Selinger15}. However, we did not leverage the presentation given
in \cite{Selinger15} for complex stabilizers in any systematic way. We
plan to explore this connection in future work, with the hope of
devising more general methods for the construction of presentations
such as the one provided here.

% --------------------------------------------------------------------
\bibliographystyle{eptcs} \bibliography{realstab}

% --------------------------------------------------------------------
\appendix

\section{An Alternative Reduced Set of Relations}
\label{app:relsalt}

\begin{figure}[h]
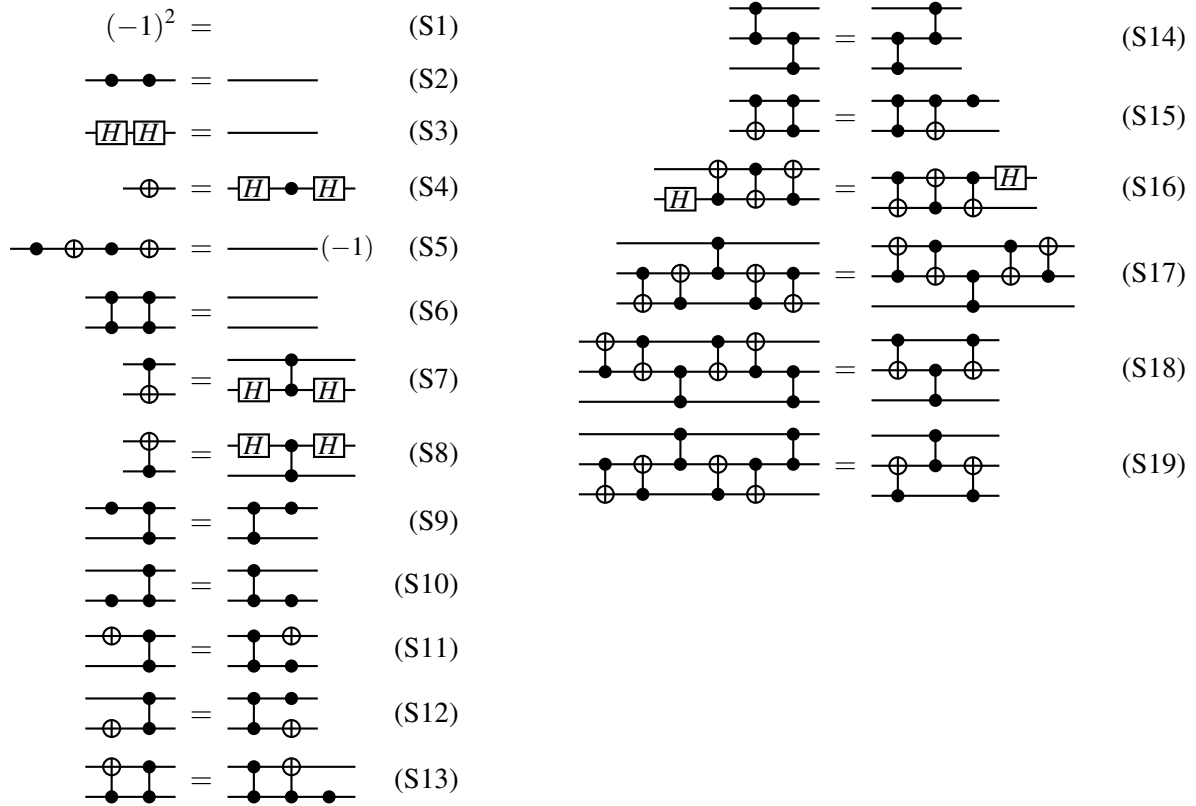

  \begingroup
  \setcounter{equation}{0}
  \renewcommand{\theequation}{S\arabic{equation}}
  \addtolength{\jot}{0.3em}
  \def\scale{0.25}
  \noindent
  \begin{minipage}[t]{0.4\textwidth}
  \begin{align}
    (-1)^2 &~=
    \\
    \m{\begin{qcircuit}[scale=\scale]
      \blackgrid{0}{4.8}{0}
      \ZGATE{1.4}{0}
      \ZGATE{3.4}{0}
    \end{qcircuit}}
    &~=\m{\begin{qcircuit}[scale=\scale]
        \blackgrid{0}{4.8}{0}
    \end{qcircuit}}
    \\
    \m{\begin{qcircuit}[scale=\scale]
      \blackgrid{0}{4.8}{0}
      \HGATE{1.4}{0}
      \HGATE{3.4}{0}
    \end{qcircuit}}
    &~=\m{\begin{qcircuit}[scale=\scale]
        \blackgrid{0}{4.8}{0}
    \end{qcircuit}}
    \\
    \m{\begin{qcircuit}[scale=\scale]
      \blackgrid{0}{2.8}{0}
      \XGATE{1.4}{0}
    \end{qcircuit}}
    &~=\m{\begin{qcircuit}[scale=\scale]
        \blackgrid{0}{6.8}{0}
        \HGATE{1.4}{0}
        \ZGATE{3.4}{0}
        \HGATE{5.4}{0}    
      \end{qcircuit}}
    \\
    \m{\begin{qcircuit}[scale=\scale]
      \blackgrid{0}{8.8}{0}
      \ZGATE{1.4}{0}
      \XGATE{3.4}{0}
      \ZGATE{5.4}{0}
      \XGATE{7.4}{0}
    \end{qcircuit}}
    &~=\m{\begin{qcircuit}[scale=\scale]
        \LABEL{6.2}{0}{~$(-1)$}
        \blackgrid{0}{4.8}{0}
    \end{qcircuit}}
    \\
    \m{\begin{qcircuit}[scale=\scale]
      \blackgrid{0}{4.8}{0,1.6}
      \CZGATE{3.4}{1.6}{0}
      \CZGATE{1.4}{1.6}{0}
    \end{qcircuit}}
    &~=\m{\begin{qcircuit}[scale=\scale]
        \blackgrid{0}{4.8}{0,1.6}
    \end{qcircuit}}
    \\
    \m{\begin{qcircuit}[scale=\scale]
      \blackgrid{0}{2.8}{0,1.6}
      \CXGATE{1.4}{1.6}{0}
    \end{qcircuit}}
    &~=\m{\begin{qcircuit}[scale=\scale]
    \blackgrid{0}{6.8}{0,1.6}
    \HGATE{1.4}{0}
    \CZGATE{3.4}{0}{1.6}
    \HGATE{5.4}{0}
    \end{qcircuit}}
    \\
    \m{\begin{qcircuit}[scale=\scale]
      \blackgrid{0}{2.8}{0,1.6}
      \XCGATE{1.4}{1.6}{0}
    \end{qcircuit}}
    &~=\m{\begin{qcircuit}[scale=\scale]
    \blackgrid{0}{6.8}{0,1.6}
    \HGATE{1.4}{1.6}
    \CZGATE{3.4}{0}{1.6}
    \HGATE{5.4}{1.6}
    \end{qcircuit}}
    \\
    \m{\begin{qcircuit}[scale=\scale]
      \blackgrid{0}{4.8}{0,1.6}
      \CZGATE{3.4}{1.6}{0}
      \ZGATE{1.4}{1.6}
    \end{qcircuit}}
    &~=\m{\begin{qcircuit}[scale=\scale]
    \blackgrid{0}{4.8}{0,1.6}
    \CZGATE{1.4}{1.6}{0}
    \ZGATE{3.4}{1.6}
    \end{qcircuit}}
    \\
    \m{\begin{qcircuit}[scale=\scale]
      \blackgrid{0}{4.8}{0,1.6}
      \CZGATE{3.4}{1.6}{0}
      \ZGATE{1.4}{0}
    \end{qcircuit}}
    &~=\m{\begin{qcircuit}[scale=\scale]
    \blackgrid{0}{4.8}{0,1.6}
    \CZGATE{1.4}{1.6}{0}
    \ZGATE{3.4}{0}
    \end{qcircuit}}
    \\
    \m{\begin{qcircuit}[scale=\scale]
      \blackgrid{0}{4.8}{0,1.6}
      \XGATE{1.4}{1.6}
      \CZGATE{3.4}{1.6}{0}
    \end{qcircuit}}
    &~=\m{\begin{qcircuit}[scale=\scale]
    \blackgrid{0}{4.8}{0,1.6}
    \CZGATE{1.4}{1.6}{0}
    \XGATE{3.4}{1.6}
    \ZGATE{3.4}{0}
    \end{qcircuit}}
    \\
    \m{\begin{qcircuit}[scale=\scale]
      \blackgrid{0}{4.8}{0,1.6}
      \XGATE{1.4}{0}
      \CZGATE{3.4}{1.6}{0}
    \end{qcircuit}}
    &~=\m{\begin{qcircuit}[scale=\scale]
    \blackgrid{0}{4.8}{0,1.6}
    \CZGATE{1.4}{1.6}{0}
    \XGATE{3.4}{0}
    \ZGATE{3.4}{1.6}
    \end{qcircuit}}
    \\
    \m{\begin{qcircuit}[scale=\scale]
      \blackgrid{0}{4.8}{0,1.6}
      \XCGATE{1.4}{1.6}{0}
      \CZGATE{3.4}{1.6}{0}
    \end{qcircuit}}
    &~=\m{\begin{qcircuit}[scale=\scale]
    \blackgrid{0}{6.8}{0,1.6}
    \CZGATE{1.4}{1.6}{0}
    \XCGATE{3.4}{1.6}{0}
    \ZGATE{5.4}{0}
    \end{qcircuit}}
  \end{align}
  \end{minipage}%
  \hspace{0.05\textwidth}
  \begin{minipage}[t]{0.55\textwidth}
  \begin{align}
    \m{\begin{qcircuit}[scale=\scale]
      \blackgrid{0}{4.8}{0,1.6,3.2}
      \CZGATE{1.4}{3.2}{1.6}
      \CZGATE{3.4}{1.6}{0}
    \end{qcircuit}}
    &~=\m{\begin{qcircuit}[scale=\scale]
    \blackgrid{0}{4.8}{0,1.6,3.2}
    \CZGATE{3.4}{3.2}{1.6}
    \CZGATE{1.4}{1.6}{0}
    \end{qcircuit}}
    \\
    \m{\begin{qcircuit}[scale=\scale]
      \blackgrid{0}{4.8}{0,1.6}
      \CXGATE{1.4}{1.6}{0}
      \CZGATE{3.4}{1.6}{0}
    \end{qcircuit}}
    &~=\m{\begin{qcircuit}[scale=\scale]
    \blackgrid{0}{6.8}{0,1.6}
    \CZGATE{1.4}{1.6}{0}
    \CXGATE{3.4}{1.6}{0}
    \ZGATE{5.4}{1.6}
    \end{qcircuit}}
    \\
    \m{\begin{qcircuit}[scale=\scale]
      \blackgrid{0}{8.8}{0,1.6}
      \HGATE{1.4}{0}
      \XCGATE{3.4}{1.6}{0}
      \CXGATE{5.4}{1.6}{0}
      \XCGATE{7.4}{1.6}{0}
    \end{qcircuit}}
    &~=\m{\begin{qcircuit}[scale=\scale]
    \blackgrid{0}{8.8}{0,1.6}
    \CXGATE{1.4}{1.6}{0}
    \XCGATE{3.4}{1.6}{0}
    \CXGATE{5.4}{1.6}{0}
    \HGATE{7.4}{1.6}    
    \end{qcircuit}}
    \\
    \m{\begin{qcircuit}[scale=\scale]
      \blackgrid{0}{10.8}{0,1.6,3.2}
      \CXGATE{1.4}{1.6}{0}
      \XCGATE{3.4}{1.6}{0}
      \CZGATE{5.4}{3.2}{1.6}
      \XCGATE{7.4}{1.6}{0}
      \CXGATE{9.4}{1.6}{0}
    \end{qcircuit}}
    &~=\m{\begin{qcircuit}[scale=\scale]
    \blackgrid{0}{10.8}{0,1.6,3.2}  
    \XCGATE{1.4}{3.2}{1.6}
    \CXGATE{3.4}{3.2}{1.6}
    \CZGATE{5.4}{1.6}{0}
    \CXGATE{7.4}{3.2}{1.6}
    \XCGATE{9.4}{3.2}{1.6}
    \end{qcircuit}}
    \\ 
    \m{\begin{qcircuit}[scale=\scale]
      \blackgrid{0}{12.8}{0,1.6,3.2}
      \XCGATE{1.4}{3.2}{1.6}
      \CXGATE{3.4}{3.2}{1.6}
      \CZGATE{5.4}{1.6}{0}
      \CXGATE{7.4}{3.2}{1.6}
      \XCGATE{9.4}{3.2}{1.6}
      \CZGATE{11.4}{1.6}{0}
    \end{qcircuit}}
    &~=\m{\begin{qcircuit}[scale=\scale]
    \blackgrid{0}{6.8}{0,1.6,3.2}
    \CXGATE{1.4}{3.2}{1.6}
    \CZGATE{3.4}{1.6}{0}
    \CXGATE{5.4}{3.2}{1.6}
    \end{qcircuit}}
    \\ 
    \m{\begin{qcircuit}[scale=\scale]
      \blackgrid{0}{12.8}{0,1.6,3.2}
      \CXGATE{1.4}{1.6}{0}
      \XCGATE{3.4}{1.6}{0}
      \CZGATE{5.4}{3.2}{1.6}
      \XCGATE{7.4}{1.6}{0}
      \CXGATE{9.4}{1.6}{0}
      \CZGATE{11.4}{3.2}{1.6}
    \end{qcircuit}}
    &~=\m{\begin{qcircuit}[scale=\scale]
    \blackgrid{0}{6.8}{0,1.6,3.2}
    \XCGATE{1.4}{1.6}{0}
    \CZGATE{3.4}{3.2}{1.6}
    \XCGATE{5.4}{1.6}{0}
    \end{qcircuit}}
  \end{align}
  \end{minipage}
  \endgroup
  \bigskip
  
  \caption{An alternative set of reduced relations for real
    stabilizers.}
  \label{fig:relsalt}
\end{figure}

\section{Typed Relations}
\label{app:rels}

\vspace{-2ex}
\begin{figure}[H]
\[
\arraycolsep=10.5pt\def\arraystretch{2}
\begin{array}{l}
\mp{0.26}{\begin{qcircuit}[scale=0.35]
  \blackgrid{0}{3.6}{0}
  \greengrid{3.6}{5}{0}
  \ZGATE{1.4}{0}
  \AONE{3.6}{0}
\end{qcircuit}}
=
\mp{0.26}{\begin{qcircuit}[scale=0.35]
  \blackgrid{0}{1.4}{0}
  \greengrid{1.4}{5}{0}
  \AONE{1.4}{0}
  \ZGATE{3.6}{0}  
\end{qcircuit}}
\\
\mp{0.26}{\begin{qcircuit}[scale=0.35]
  \blackgrid{0}{3.6}{0}
  \greengrid{3.6}{5}{0}
  \ZGATE{1.4}{0}
  \ATWO{3.6}{0}
\end{qcircuit}}
=
\mp{0.26}{\begin{qcircuit}[scale=0.35]
  \blackgrid{0}{1.4}{0}
  \greengrid{1.4}{5}{0}
  \ATWO{1.4}{0}
  \XGATE{3.6}{0}  
\end{qcircuit}}
\\
\mp{0.26}{\begin{qcircuit}[scale=0.35]
  \blackgrid{0}{3.6}{0}
  \redgrid{3.6}{5}{0}
  \ZGATE{1.4}{0}
  \ATHREE{3.6}{0}
\end{qcircuit}}
=
\mp{0.26}{\begin{qcircuit}[scale=0.35]
  \blackgrid{0}{1.4}{0}
  \redgrid{1.4}{5}{0}
  \ATHREE{1.4}{0}
  \ZGATE{3.6}{0}  
\end{qcircuit}}
\end{array}
\begin{array}{l}
\mp{0.26}{\begin{qcircuit}[scale=0.35]
  \blackgrid{0}{3.6}{0}
  \greengrid{3.6}{5}{0}
  \HGATE{1.4}{0}
  \AONE{3.6}{0}
\end{qcircuit}}
=
\mp{0.26}{\begin{qcircuit}[scale=0.35]
  \blackgrid{0}{1.4}{0}
  \greengrid{1.4}{2.8}{0}
  \ATWO{1.4}{0}
\end{qcircuit}}
\\
\mp{0.26}{\begin{qcircuit}[scale=0.35]
  \blackgrid{0}{3.6}{0}
  \greengrid{3.6}{5}{0}
  \HGATE{1.4}{0}
  \ATWO{3.6}{0}
\end{qcircuit}}
=
\mp{0.26}{\begin{qcircuit}[scale=0.35]
  \blackgrid{0}{1.4}{0}
  \greengrid{1.4}{2.8}{0}
  \AONE{1.4}{0}
\end{qcircuit}}
\\
\mp{0.26}{\begin{qcircuit}[scale=0.35]
  \blackgrid{0}{3.6}{0}
  \redgrid{3.6}{5}{0}
  \HGATE{1.4}{0}
  \ATHREE{3.6}{0}
\end{qcircuit}}
=
\mp{0.26}{\begin{qcircuit}[scale=0.35]
  \blackgrid{0}{1.4}{0}
  \redgrid{1.4}{5}{0}
  \ATWO{1.4}{0}
  \HGATE{3.6}{0}  
\end{qcircuit}}
\end{array}%
\]
\caption{Rewrite rules for normal forms, part I.\label{fig:partone}}
\end{figure}

\begin{figure}
\[
%
\]
\caption{Rewrite rules for normal forms, part II.\label{fig:part2}}
\end{figure}

\begin{figure}
\[
\arraycolsep=10.5pt\def\arraystretch{3}
\begin{array}{l}
\mp{0.26}{\begin{qcircuit}[scale=0.35]
    \blackgrid{0}{5}{0}
    \greengrid{0}{4}{1.6}
    \blackgrid{4}{5}{1.6}
    \blackgrid{0}{4}{3.2}
    \greengrid{4}{5}{3.2}
    \CZGATE{1.4}{1.6}{0}
    \BONE{3.6}{3.2}{1.6}
\end{qcircuit}}
=
\mp{0.26}{\begin{qcircuit}[scale=0.35]
    \blackgrid{0}{20.4}{0}
    \greengrid{0}{2}{1.6}
    \blackgrid{2}{20.4}{1.6}
    \blackgrid{0}{2}{3.2}
    \greengrid{2}{20.4}{3.2}
    \BONE{1.4}{3.2}{1.6}
    \CZGATE{3.6}{3.2}{1.6}
    \HGATE{5.8}{1.6}
    \CZGATE{8}{1.6}{0}
    \HGATE{10.2}{1.6}
    \CZGATE{12.4}{3.2}{1.6}
    \HGATE{14.6}{1.6}
    \CZGATE{16.8}{1.6}{0}
    \HGATE{19}{1.6}
\end{qcircuit}}
\\
\mp{0.26}{\begin{qcircuit}[scale=0.35]
    \blackgrid{0}{5}{0}
    \greengrid{0}{4}{1.6}
    \blackgrid{4}{5}{1.6}
    \blackgrid{0}{4}{3.2}
    \greengrid{4}{5}{3.2}
    \CZGATE{1.4}{1.6}{0}
    \BTWO{3.6}{3.2}{1.6}
\end{qcircuit}}
=
\mp{0.26}{\begin{qcircuit}[scale=0.35]
    \blackgrid{0}{13.8}{0}
    \greengrid{0}{2}{1.6}
    \blackgrid{2}{13.8}{1.6}
    \blackgrid{0}{2}{3.2}
    \greengrid{2}{13.8}{3.2}
    \BTWO{1.4}{3.2}{1.6}
    \CZGATE{3.6}{3.2}{1.6}
    \HGATE{5.8}{1.6}
    \CZGATE{8}{1.6}{0}
    \HGATE{10.2}{1.6}
    \CZGATE{12.4}{3.2}{1.6}
\end{qcircuit}}
\\
\mp{0.26}{\begin{qcircuit}[scale=0.35]
    \blackgrid{0}{5}{0}
    \greengrid{0}{4}{1.6}
    \blackgrid{4}{5}{1.6}
    \blackgrid{0}{4}{3.2}
    \greengrid{4}{5}{3.2}
    \CZGATE{1.4}{1.6}{0}
    \BTHREE{3.6}{3.2}{1.6}
\end{qcircuit}}
=
\mp{0.26}{\begin{qcircuit}[scale=0.35]
    \blackgrid{0}{13.8}{0}
    \greengrid{0}{2}{1.6}
    \blackgrid{2}{13.8}{1.6}
    \blackgrid{0}{2}{3.2}
    \greengrid{2}{13.8}{3.2}
    \BTHREE{1.4}{3.2}{1.6}
    \CZGATE{3.6}{3.2}{1.6}
    \HGATE{5.8}{1.6}
    \CZGATE{8}{1.6}{0}
    \HGATE{10.2}{1.6}
    \CZGATE{12.4}{3.2}{1.6}
\end{qcircuit}}
\\
\end{array}
\]
\caption{Rewrite rules for normal forms, part III.\label{fig:part3}}
\end{figure}
\begin{figure}
\[
\arraycolsep=20.5pt\def\arraystretch{3}
\begin{array}{ll}
\mp{0.26}{\begin{qcircuit}[scale=0.3]
    \greengrid{0}{4}{0}
    \blackgrid{4}{7.2}{0}
    \blackgrid{0}{4}{1.6}
    \greengrid{4}{6}{1.6}
    \blackgrid{6}{7.2}{1.6}
    \blackgrid{0}{6}{3.2}
    \greengrid{6}{7.2}{3.2}
    \CZGATE{1.4}{3.2}{1.6}
    \BONE{3.6}{1.6}{0}
    \BONE{5.8}{3.2}{1.6}
\end{qcircuit}}
=
\mp{0.26}{\begin{qcircuit}[scale=0.3]
    \greengrid{0}{1.4}{0}
    \blackgrid{1.4}{11.6}{0}
    \blackgrid{0}{1.4}{1.6}
    \greengrid{1.4}{3.6}{1.6}
    \blackgrid{3.6}{11.6}{1.6}
    \blackgrid{0}{3.6}{3.2}
    \greengrid{3.6}{11.6}{3.2}
    \BONE{1.4}{1.6}{0}
    \BONE{3.6}{3.2}{1.6}
    \HGATE{5.8}{0}
    \HGATE{5.8}{1.6}
    \CZGATE{8}{1.6}{0}
    \HGATE{10.2}{0}
    \HGATE{10.2}{1.6}
\end{qcircuit}}
&
\mp{0.26}{\begin{qcircuit}[scale=0.3]
   \greengrid{0}{4}{0}
    \blackgrid{4}{7.2}{0}
    \blackgrid{0}{4}{1.6}
    \greengrid{4}{6}{1.6}
    \blackgrid{6}{7.2}{1.6}
    \blackgrid{0}{6}{3.2}
    \greengrid{6}{7.2}{3.2}
    \CZGATE{1.4}{3.2}{1.6}
    \BONE{3.6}{1.6}{0}
    \BTWO{5.8}{3.2}{1.6}
\end{qcircuit}}
=
\mp{0.26}{\begin{qcircuit}[scale=0.3]
    \greengrid{0}{1.4}{0}
    \blackgrid{1.4}{11.6}{0}
    \blackgrid{0}{1.4}{1.6}
    \greengrid{1.4}{3.6}{1.6}
    \blackgrid{3.6}{11.6}{1.6}
    \blackgrid{0}{3.6}{3.2}
    \greengrid{3.6}{11.6}{3.2}
    \BTHREE{1.4}{1.6}{0}
    \BTWO{3.6}{3.2}{1.6}
    \HGATE{5.8}{0}
    \CZGATE{8}{1.6}{0}
    \HGATE{10.2}{0}
\end{qcircuit}}
\\
\end{array}
\]
\caption{Rewrite rules for normal forms, part IV.\label{fig:part4}}
\end{figure}

\begin{figure}
\[
\arraycolsep=20.5pt\def\arraystretch{3}
\begin{array}{ll}
\mp{0.26}{\begin{qcircuit}[scale=0.3]
    \redgrid{0}{4}{0}
    \blackgrid{4}{7.2}{0}
    \blackgrid{0}{4}{1.6}
    \redgrid{4}{6}{1.6}
    \blackgrid{6}{7.2}{1.6}
    \blackgrid{0}{6}{3.2}
    \redgrid{6}{7.2}{3.2}
    \CZGATE{1.4}{3.2}{1.6}
    \BFIVE{3.6}{1.6}{0}
    \BFIVE{5.8}{3.2}{1.6}
\end{qcircuit}}
=
\mp{0.26}{\begin{qcircuit}[scale=0.3]
    \redgrid{0}{1.4}{0}
    \blackgrid{1.4}{11.6}{0}
    \blackgrid{0}{1.4}{1.6}
    \redgrid{1.4}{3.6}{1.6}
    \blackgrid{3.6}{11.6}{1.6}
    \blackgrid{0}{3.6}{3.2}
    \redgrid{3.6}{11.6}{3.2}
    \BFIVE{1.4}{1.6}{0}
    \BFIVE{3.6}{3.2}{1.6}
    \HGATE{5.8}{0}
    \HGATE{5.8}{1.6}
    \CZGATE{8}{1.6}{0}
    \HGATE{10.2}{0}
    \HGATE{10.2}{1.6}
\end{qcircuit}}
&
\mp{0.26}{\begin{qcircuit}[scale=0.3]
    \redgrid{0}{4}{0}
    \blackgrid{4}{7.2}{0}
    \blackgrid{0}{4}{1.6}
    \redgrid{4}{6}{1.6}
    \blackgrid{6}{7.2}{1.6}
    \blackgrid{0}{6}{3.2}
    \redgrid{6}{7.2}{3.2}
    \CZGATE{1.4}{3.2}{1.6}
    \BFIVE{3.6}{1.6}{0}
    \BSIX{5.8}{3.2}{1.6}
\end{qcircuit}}
=
\mp{0.26}{\begin{qcircuit}[scale=0.3]
    \redgrid{0}{1.4}{0}
    \blackgrid{1.4}{11.6}{0}
    \blackgrid{0}{1.4}{1.6}
    \redgrid{1.4}{3.6}{1.6}
    \blackgrid{3.6}{11.6}{1.6}
    \blackgrid{0}{3.6}{3.2}
    \redgrid{3.6}{11.6}{3.2}
    \BSEV{1.4}{1.6}{0}
    \BSIX{3.6}{3.2}{1.6}
    \HGATE{5.8}{0}
    \CZGATE{8}{1.6}{0}
    \HGATE{10.2}{0}
\end{qcircuit}}
\\
\end{array}
\]
\caption{Rewrite rules for normal forms, part V.\label{fig:part5}}
\end{figure}
\begin{figure}
\[
\arraycolsep=10.5pt\def\arraystretch{3}
%
\]
\caption{Rewrite rules for normal forms, part VI.\label{fig:part6}}
\end{figure}
\begin{figure}
\[
\arraycolsep=10.5pt\def\arraystretch{4}
\begin{array}{lll}
\mp{0.43}{\begin{qcircuit}[scale=0.33]
  \blackgrid{0}{7.2}{0,1.6,3.2}
  \CZGATE{1.4}{1.6}{0}
  \DONE{3.6}{3.2}{1.6}
  \DONE{5.8}{1.6}{0}
\end{qcircuit}}
=
\mp{0.43}{\begin{qcircuit}[scale=0.33]
  \blackgrid{0}{11.6}{0,1.6,3.2}
  \DONE{1.4}{3.2}{1.6}
  \DONE{3.6}{1.6}{0}
  \HGATE{5.8}{3.2}
  \HGATE{5.8}{1.6}
  \CZGATE{8}{3.2}{1.6}
  \HGATE{10.2}{3.2}
  \HGATE{10.2}{1.6}  
\end{qcircuit}}
& \ \ \ \ \ \ \ \
\mp{0.43}{\begin{qcircuit}[scale=0.33]
  \blackgrid{0}{7.2}{0,1.6,3.2}
  \CZGATE{1.4}{1.6}{0}
  \DONE{3.6}{3.2}{1.6}
  \DTWO{5.8}{1.6}{0}
\end{qcircuit}}
=
\mp{0.43}{\begin{qcircuit}[scale=0.33]
  \blackgrid{0}{11.6}{0,1.6,3.2}
  \DTHREE{1.4}{3.2}{1.6}
  \DTWO{3.6}{1.6}{0}
  \HGATE{5.8}{3.2}
  \CZGATE{8}{3.2}{1.6}
  \HGATE{10.2}{3.2}
\end{qcircuit}}
\\
\mp{0.43}{\begin{qcircuit}[scale=0.33]
  \blackgrid{0}{7.2}{0,1.6,3.2}
  \CZGATE{1.4}{1.6}{0}
  \DONE{3.6}{3.2}{1.6}
  \DTHREE{5.8}{1.6}{0}
\end{qcircuit}}
=
\mp{0.43}{\begin{qcircuit}[scale=0.33]
  \blackgrid{0}{11.6}{0,1.6,3.2}
  \DONE{1.4}{3.2}{1.6}
  \DTHREE{3.6}{1.6}{0}
  \HGATE{5.8}{3.2}
  \HGATE{5.8}{1.6}
  \CZGATE{8}{3.2}{1.6}
  \HGATE{10.2}{3.2}
  \HGATE{10.2}{1.6}
\end{qcircuit}}
& \ \ \ \ \ \ \ \
\mp{0.43}{\begin{qcircuit}[scale=0.33]
  \blackgrid{0}{7.2}{0,1.6,3.2}
  \CZGATE{1.4}{1.6}{0}
  \DONE{3.6}{3.2}{1.6}
  \DFOUR{5.8}{1.6}{0}
\end{qcircuit}}
=
\mp{0.43}{\begin{qcircuit}[scale=0.33]
  \blackgrid{0}{11.6}{0,1.6,3.2}
  \DTHREE{1.4}{3.2}{1.6}
  \DFOUR{3.6}{1.6}{0}
  \HGATE{5.8}{3.2}
  \CZGATE{8}{3.2}{1.6}
  \HGATE{10.2}{3.2}
\end{qcircuit}}
\\
\end{array}
\]
\caption{Rewrite rules for normal forms, part VII.\label{fig:part7}}
\end{figure}

\begin{figure}
\[
\begin{array}{l}
\begin{qcircuit}[scale=0.3]
    \LABEL{4.2}{0}{$=$}
    \blackgrid{0}{2.8}{0}
    \blackgrid{5.6}{7}{0}
    \greengrid{7}{8.4}{0}
    \blackgrid{8.4}{12.8}{0}
    
    \AONE{7}{0}
    \CONE{9.2}{0}
    \EONE{11.4}{0}
\end{qcircuit}
\\ \\
\begin{qcircuit}[scale=0.3]
    \blackgrid{0}{2.8}{1.6}
    \greengrid{0}{1.4}{0}
    \blackgrid{1.4}{2.8}{0}
    
    \blackgrid{5.6}{7}{1.6}
    \greengrid{7}{9.2}{1.6}
    \blackgrid{9.2}{12.8}{1.6}
    
    \greengrid{5.6}{7}{0}
    \blackgrid{7}{12.8}{0}
    
    \CONE{1.4}{0}
    \BONE{7}{1.6}{0}
    \CONE{9.2}{1.6}
    \DONE{11.4}{1.6}{0}
    
\end{qcircuit}
\end{array}
\]
\caption{Rewrite rules for normal forms, part VIII.\label{fig:part8}}
\end{figure}

\clearpage

\section{Proofs of Propositions 4.11, 4.12, 4.13, 4.14, and 4.15}
\label{app:nfproofs}

\textbf{Proposition 4.11}
  Let P be a $n$-qubit Pauli operator, with $P = P_{1} \otimes P_{2}
  \otimes \cdots \otimes P_{n}$, $P^2 = I$, and $P \neq \pm I$. Then
  there exists a unique $Z$-circuit $L$ such that $L \bullet P = Z
  \otimes I \otimes \cdots \otimes I$.

\begin{proof}
Since $P \neq \pm I$, there is an index $m$ such that $P_{m} \neq \pm
I$. Let $m$ be the largest such index. Then $P_{m} = \pm X,P_{m} = \pm
Z,$ or $P_{m} = \pm XZ$. With this, we consider the following diagram.
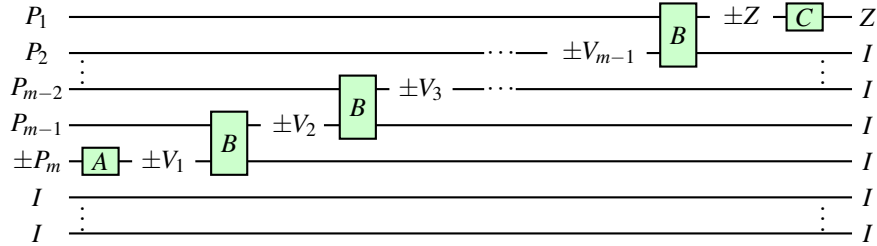
\[
\begin{qcircuit}[scale=0.3]
    \LABEL{-1.4}{0}{$\pm P_{m}$}
    \LABEL{-1.4}{1.6}{$P_{m-1}$}
    \LABEL{-1.4}{3.2}{$P_{m-2}$}
    \LABEL{-1.4}{4.8}{$P_{2}$}
    \LABEL{-1.4}{6.4}{$P_{1}$}
    \LABEL{-1.4}{-1.6}{$I$}
    \LABEL{-1.4}{-3.2}{$I$}
    \LABEL{0.6}{-2.2}{$\vdots$}
    \LABEL{33.4}{-2.2}{$\vdots$}
    
    \LABEL{35.4}{6.4}{$Z$}
    \LABEL{35.4}{4.8}{$I$}
    \LABEL{35.4}{3.2}{$I$}
    \LABEL{35.4}{1.6}{$I$}
    \LABEL{35.4}{0}{$I$}
    \LABEL{35.4}{-1.6}{$I$}
    \LABEL{35.4}{-3.2}{$I$}
    
    \LABEL{0.6}{4.3}{$\vdots$}
    \LABEL{33.4}{4.3}{$\vdots$}
    \LABEL{4.2}{0}{$\pm V_{1}$}
    \LABEL{10}{1.6}{$\pm V_{2}$}
    \LABEL{15.6}{3.2}{$\pm V_{3}$}
    \LABEL{19.2}{3.2}{$\cdots$}
    \LABEL{19.2}{4.8}{$\cdots$}
    \LABEL{23.5}{4.8}{$\pm V_{m-1}$}
    \LABEL{29.8}{6.4}{$\pm Z$}
    
    \blackgrid{0}{1.4}{0}
    \blackgrid{1.4}{2.8}{0}
    \blackgrid{5.6}{7}{0}
    \blackgrid{7}{34.7}{0}
    
    \blackgrid{0}{7.1}{1.6}
    \blackgrid{7.1}{8.5}{1.6}
    \blackgrid{11.3}{12.7}{1.6}
    \blackgrid{12.7}{34.7}{1.6}
    
    \blackgrid{0}{18.4}{4.8}
    \blackgrid{19.8}{21.2}{4.8}
    \blackgrid{25.6}{27}{4.8}
    \blackgrid{27}{34.7}{4.8}
    
    \blackgrid{0}{12.8}{3.2}
    \blackgrid{12.8}{14.2}{3.2}
    \blackgrid{17}{18.4}{3.2}
    \blackgrid{19.8}{34.7}{3.2}
    
    \blackgrid{0}{27}{6.4}
    \greengrid{27}{28.4}{6.4}
    \greengrid{31.2}{32.6}{6.4}
    \blackgrid{32.6}{34.7}{6.4}
    
    \blackgrid{0}{34.7}{-1.6}
    \blackgrid{0}{34.7}{-3.2}
    
    \GENERICORANGE{1.4}{0}{0}{$A$}
    \GENERICGREEN{7.1}{1.6}{0}{$B$}
    \GENERICGREEN{12.8}{3.2}{1.6}{$B$}
    \GENERICGREEN{27}{6.4}{4.8}{$B$}
    \GENERICYELLOW{32.6}{6.4}{6.4}{$C$}
\end{qcircuit}
\]
In the above diagram, the $V_s$ are Pauli operators such that $V_{s}
\in \{Z,XZ\}$ and are determined in the following way. By
Proposition 4.7, if $P_{m}= \pm X, \pm Z$, there is a unique $A$
gate $A_{g}$ with output of single type such that $A_{g} \bullet P_{m}
= \pm Z$. If $P_{m}= \pm XZ$, there is a unique $A$ gate $A_{r}$ with
output of double type such that $A_{r} \bullet P_{m} = \pm XZ$. So the
$A$ gate is uniquely determined. Furthermore after the application of
the $A$ gate, we either have $V_1=\pm XZ$ on a wire of double type or
$V_1=\pm Z$ on a wire of single type. We will further use the actions
in Proposition 4.7 to move these $Z$ or $XZ$ Pauli operators up the
qubits.

By inspection of these actions, we see that for each choice of
$P_{m-1} \otimes V_{1}$, there is a unique $B$ gate $B_{j}$ such that
$B_{j} \bullet P_{m-1} \otimes V_{1} = V_{2} \otimes I$ and $V_{2} =
Z$ or $V_{2} = XZ$. If $V_2=Z$, the output wire is of single type, and
if $V_2=XZ$, the output wire is of double type. We can continue this
process up to the top qubit and this will produce a $Z$-circuit if we
can ensure that the top output wire is of single type. Since $P^{2}
=I$, there are evenly many indices $i$ such that $P_{i}=XZ$, which are
in effect cancelled out by an application of $B_{8}$, switching back to
$Z$ along a wire of single type. Thus we will always end up
constructing a circuit $C$ such that $C \bullet P = \pm Z \otimes I
\otimes \ldots \otimes I$, to which there is a unique $C$-gate $C_{k}$
such that $C_{k} \bullet \pm Z = Z$. This completes the proof of
existence. 

Note that every choice of gate is unique with respect to kind and
type. If our normal form was constructed the same way in the absence
of types, uniqueness with respect to kind would be sufficient for a
unique $Z$-circuit. Here, with uniqueness with respect to kind and
type, we must prove that no two $Z$-circuits describing an action as
above can have different typing schemes. Consider two $Z$-circuits $C$
and $D$ that correspond to the diagram above, such that $C \bullet P =
D \bullet P = Z \otimes I \otimes \ldots \otimes I$. We now show that
they have the same typing schemes. Note that both $A$ gates in $C$ and
$D$ must satisfy $A \bullet P_{m} = \pm V_{1}$, where $V_{1} = \pm Z,
\pm XZ$. $A_{2}$ is the only $A$ gate such that $A \bullet \pm X = \pm
Z$, and the equations $A \bullet \pm Z = \pm Z$ and $A \bullet XZ =
\pm XZ$ both have two $A$ gates with these properties, $A_{1}$ and
$A_{3}$. Both of these gates are different with respect to output
type, but represent the same actions. When an $A_{1}$ is chosen as the
$A$ gate, there is an even number of gates from the set
$\{B_{4},B_{8}\}$ which appear to its right, as these $B$ gates switch
the type up the ladder. If $A_{3}$ is chosen as the $A$ gate, then
there is an odd number of gates from the set $\{B_{4},B_{8}\}$ which
appear to its right.  Thus it is not possible for both circuits $C$
and $D$ to start with the different $A$ gates $A_{1}$ and $A_{3}$
respectively, as it is not possible for both resulting circuits to
have $C \bullet P = D \bullet P = Z \otimes I \ldots \otimes I$ with a
different number of occurrences of a given local action. Hence, $C$
and $D$ share the same $A$ gate, and have the same starting type. Note
that if the input type is given, there are four choices of possible
local actions of $B \bullet P_{m-j} \otimes V_{j} = V_{j+1} \otimes
I$, corresponding to $B_{1},B_{2},B_{3},B_{4}$ in the case of a single
type, and $B_{5},B_{6},B_{7},B_{8}$ in the case of a double
type. Since the output type of $A$ is given, and we must satisfy the
equations $B \bullet P_{m-j} \otimes V_{j} = V_{j+1} \otimes I$, there
are four choices for four possibilities at each choice of $B$, which
all describe different actions. Here we see that with a shared $A$
gate, both $Z$-circuits $C$ and $D$ must also have the same $B$ gates,
and thus the same typing scheme, ending in a single type, with the
corresponding unique choice of a $C$ gate such that $C \bullet \pm Z =
Z$. Hence we have that $C$ and $D$ have the same typing scheme. Since
the typing schemes must be the same, all local actions must
coincide. Hence the two $Z$-circuits are equal. This proves
uniqueness.
\end{proof}
\noindent\textbf{Proposition 4.12}
  Let $Q$ be an $n$-qubit Pauli operator with $Q= Q_1\otimes
  Q_2\otimes \cdots \otimes Q_n$, $Q^{2}=I$, $Q \neq \pm I$, and $Q$
  anticommutes with $Z \otimes I \otimes \cdots \otimes I$. Then there
  exists a unique $X$-circuit $M$ such that $M \bullet Q = I \otimes
  \cdots \otimes I \otimes X$.

\begin{proof}
  Since $Q$ anticommutes with $Z \otimes I \otimes \cdots \otimes I$,
  we have $Q_{1} = \pm XZ$ or $Q_{1} = \pm X$. With this, consider the
  diagram
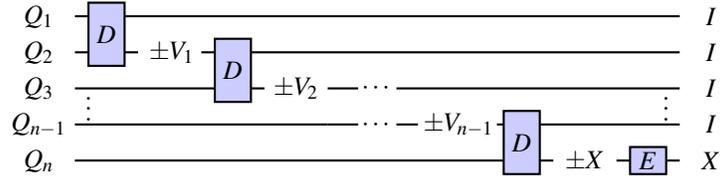
\[
\begin{qcircuit}[scale=0.3]
    \LABEL{0.6}{2.65}{$\vdots$}
    \LABEL{26.2}{2.65}{$\vdots$}
  
    \blackgrid{0}{26.8}{6.4}
    
    \blackgrid{0}{2.8}{4.8}
    \blackgrid{5.6}{26.8}{4.8}
    
    \blackgrid{0}{8.4}{3.2}
    \blackgrid{11.2}{12.8}{3.2}
    \blackgrid{14.2}{26.8}{3.2}
    
    \blackgrid{0}{11.2}{1.6}
    \blackgrid{11.2}{12.8}{1.6}
    \blackgrid{14.2}{15.2}{1.6}
    \blackgrid{18.7}{26.8}{1.6}
    
    \blackgrid{0}{21.2}{0}
    \blackgrid{24}{26.8}{0}
    
    \LABEL{17}{1.6}{$\pm V_{n-1}$}
    \LABEL{13.4}{1.6}{$\cdots$}
    \LABEL{13.4}{3.2}{$\cdots$}
    \LABEL{4.3}{4.8}{$\pm V_{1}$}
    \LABEL{9.8}{3.2}{$\pm V_{2}$}
    \LABEL{22.6}{0}{$\pm X$}
    
    \LABEL{-1.6}{6.4}{$Q_{1}$}
    \LABEL{-1.6}{4.8}{$Q_{2}$}
    \LABEL{-1.6}{3.2}{$Q_{3}$}
    \LABEL{-1.6}{1.6}{$Q_{n-1}$}
    \LABEL{-1.6}{0}{$Q_{n}$}
    
    \LABEL{28.2}{6.4}{$I$}
    \LABEL{28.2}{4.8}{$I$}
    \LABEL{28.2}{3.2}{$I$}
    \LABEL{28.2}{1.6}{$I$}
    \LABEL{28.2}{0}{$X$}
    
    \GENERICBLUE{1.4}{6.4}{4.8}{$D$}
    \GENERICBLUE{7}{4.8}{3.2}{$D$}
    \GENERICBLUE{19.8}{1.6}{0}{$D$}
    \GENERICPINK{25.4}{0}{0}{$E$}
    
\end{qcircuit}
\]
where the $V_s$ are Pauli operators such that $V_{s} \in \{X,XZ\}$,
and are determined by the $Q_{i}$ as in \textbf{Proposition 4.11}. By
\textbf{Proposition 4.7}, the $D$ gates push $X$ and $XZ$ gates down the
qubits until we encounter $XZ \otimes XZ$, at which point we apply
$D_{4}$. There is always a unique $D$ gate to perform the needed
action, leaving $V_{1} = X, XZ$. We continue the same process down to
the bottom qubit. Again since $Q^{2} = I$, by \textbf{Proposition 2.2},
there are evenly many indices $t$ such that $Q_{t} = XZ$. These
occurrences of $XZ$ get cancelled out in pairs, ensuring that we are
left with an $\pm X$ on the bottom qubit. By \textbf{Proposition 4.7}, there
is a unique $E$ gate $E_{h}$ such that $E_{h} \bullet \pm X = X$. Thus
we are left with $I \otimes \cdots \otimes I \otimes X$ and our
circuit is an $X$-circuit. Furthermore, since each gate was unique
with respect to kind, the circuit is uniquely determined.
\end{proof}

\noindent\textbf{Proposition 4.13}
  Every $X$-circuit $M$ satisfies $M \bullet (Z \otimes I \otimes
  \cdots \otimes I) = I \otimes \cdots \otimes I \otimes Z$.

\begin{proof}
The claim follows from the actions described in \textbf{Proposition 4.7} with
respect to the diagram below.
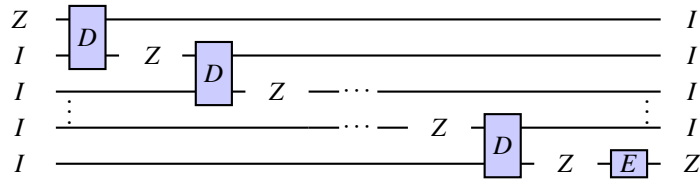
\[
\begin{qcircuit}[scale=0.3]
    \LABEL{0.6}{2.65}{$\vdots$}
    \LABEL{26.2}{2.65}{$\vdots$}
 
    \blackgrid{0}{26.8}{6.4}
    
    \blackgrid{0}{2.8}{4.8}
    \blackgrid{5.6}{26.8}{4.8}
    
    \blackgrid{0}{8.4}{3.2}
    \blackgrid{11.2}{12.8}{3.2}
    \blackgrid{14.2}{26.8}{3.2}
    
    \blackgrid{0}{11.2}{1.6}
    \blackgrid{11.2}{12.8}{1.6}
    \blackgrid{14.2}{15.6}{1.6}
    \blackgrid{18.4}{26.8}{1.6}
    
    \blackgrid{0}{21.2}{0}
    \blackgrid{24}{26.8}{0}
    
    \LABEL{17}{1.6}{$Z$}
    \LABEL{13.4}{1.6}{$\cdots$}
    \LABEL{13.4}{3.2}{$\cdots$}
    \LABEL{4.3}{4.8}{$Z$}
    \LABEL{9.8}{3.2}{$Z$}
    \LABEL{22.6}{0}{$Z$}
    
    \LABEL{-1.6}{6.4}{$Z$}
    \LABEL{-1.6}{4.8}{$I$}
    \LABEL{-1.6}{3.2}{$I$}
    \LABEL{-1.6}{1.6}{$I$}
    \LABEL{-1.6}{0}{$I$}
    
    \LABEL{28.2}{6.4}{$I$}
    \LABEL{28.2}{4.8}{$I$}
    \LABEL{28.2}{3.2}{$I$}
    \LABEL{28.2}{1.6}{$I$}
    \LABEL{28.2}{0}{$Z$}
    
    \GENERICBLUE{1.4}{6.4}{4.8}{$D$}
    \GENERICBLUE{7}{4.8}{3.2}{$D$}
    \GENERICBLUE{19.8}{1.6}{0}{$D$}
    \GENERICPINK{25.4}{0}{0}{$E$}
    
\end{qcircuit}
\]
\end{proof}
\noindent\textbf{Proposition 4.14}
  Let $P$ and $Q$ be Pauli operators such that $P^2=Q^2=I,$ $P,Q \neq
  \pm I$, and $P$ and $Q$ anticommute. Then there exists a unique pair
  of a $Z$-circuit $L$ and a $X$-circuit $M$ such that
  \begin{align*}
  ML \bullet P = I \otimes \cdots \otimes I \otimes Z & &
  and & &
  ML \bullet Q = I \otimes \cdots \otimes I \otimes X
  \end{align*}

\begin{proof}
By \textbf{Proposition 4.11}, there is a unique $Z$-circuit $L$ such that $L
\bullet P = Z \otimes I \otimes \cdots \otimes I$. Since $P$ and $Q$
both square to the identity and anticommute, so do $L \bullet P$ and
$L \bullet Q$. Thus by \textbf{Proposition 4.12}, there exists a unique
$X$-circuit $M$ such that $M \bullet (L \bullet Q) = ML \bullet Q = I
\otimes I \otimes \cdots \otimes X$ and, by \textbf{Proposition 4.13}, $ML \bullet P
= M \bullet (L \bullet P) = M \bullet (Z \otimes I \otimes \cdots
\otimes I) = I \otimes I \otimes \cdots \otimes Z$. This proves
existence. For uniqueness, we assume that $L'$ and $M'$ are two other
circuits satisfying the conditions of the proposition. Since $M'L'
\bullet P = I \otimes \cdots \otimes I \otimes Z$, and $M' \bullet (Z
\otimes I \otimes \cdots \otimes I) = I \otimes \cdots \otimes I
\otimes Z$, we can deduce that $L' \bullet P = Z \otimes I \otimes
\cdots \otimes I$. Therefore $L' = L$ by the uniqueness of
\textbf{Proposition 4.11}, and since $M' \bullet L \bullet Q = M \bullet L
\bullet Q = X \otimes I \otimes \cdots \otimes I$, we have that $M' =
M$ by the uniqueness of \textbf{Proposition 4.12}.
\end{proof}
\noindent\textbf{Proposition 4.15}
  Let $\phi : \pauli{n} \rightarrow \pauli{n}$ be an automorphism of
  the Pauli group. Then there exists a normal circuit $C$ such that
  for all $P$, $C \bullet P = \phi(P)$. Moreover, the normal form $C$
  is unique up to a scalar $\pm 1$.

\begin{proof}
We proceed by induction on $n$. When $n=0$, the Pauli operators are
the scalars $\pm 1$. Thus in this case $\phi$ is the
identity. Choosing $C = 1$, we get $C \bullet P = \phi(P)$. Uniqueness
up to scalar follows from the fact that when $n=0$, the Clifford
operators are the scalars $\pm 1$. Now suppose that our claim is true
for $n-1$ and consider the case of $n$. First we will prove existence.
Let $P = \phi^{-1}(I \otimes \ldots \otimes I \otimes Z)$ and $Q =
\phi^{-1}(I \otimes \ldots \otimes I \otimes X)$. Then $PQ =
\phi^{-1}(I \otimes \ldots \otimes I \otimes ZX)$. Since $I \otimes
\ldots \otimes I \otimes Z$ and $I \otimes \ldots \otimes I \otimes X$
anticommute, so do $P$ and $Q$. By\textbf{Proposition 4.14}, there exists a
unique $X$-circuit $M$ and a unique $Z$-circuit $L$ such that $ ML
\bullet P = I \otimes \ldots \otimes I \otimes Z = \phi(P)$ and $ML
\bullet Q = I \otimes \ldots \otimes I \otimes X = \phi(Q)$. We now
define a new automorphism $\phi ': \pauli{n} \rightarrow \pauli{n}$ by
\[
  \phi '(U) = \phi((ML)^{-1} \bullet U)
\]
for all $n$-qubit Pauli operators $U$. Note that $I \otimes \cdots
\otimes I \otimes Z$, $I \otimes \cdots \otimes I \otimes X$ and $I
\otimes \dots \otimes I \otimes ZX$ are all fixed points of $\phi'$,
since
\begin{align*}
\phi '(I \otimes \dots \otimes I \otimes Z) &= \phi((ML)^{-1} \bullet
(I \otimes \dots \otimes I \otimes Z) \\ &= \phi((ML)^{-1} \bullet
(ML) \bullet P) = \phi(P) = I \otimes \dots \otimes I \otimes Z
\end{align*}
and similarly for $I \otimes \dots \otimes I \otimes X$. Now let $R$
be an ($n-1$)-qubit Pauli operator. Since $R \otimes I$ commutes with
$I \otimes \dots \otimes I \otimes Z$ and $I \otimes \dots \otimes I
\otimes X$, the same is true of $\phi '(R \otimes I)$. Hence $\phi '(R
\otimes I) = V \otimes I$ for some $V \in \pauli{n-1}$.  It follows
that there exists an automorphism $\phi '': \pauli{n-1} \rightarrow
\pauli{n-1}$ such that, for every $R \in \pauli{n-1}$, $\phi '(R
\otimes I) = \phi ''(R) \otimes I$. Since $I \otimes \dots \otimes I
\otimes Z$ and $I \otimes \dots \otimes I \otimes X$ are fixed points
of $\phi '$, we then have $\phi ' = \phi '' \otimes I$.

By the induction hypothesis, there exists a normal $n-1$ qubit
Clifford circuit $C'$ such that for all $R \in \pauli{n-1}$, $C'
\bullet R = \phi ''(R)$. Let $C = (C' \otimes I)ML$. Since $ML \bullet
U = (\phi ')^{-1}(\phi(U))$, we see that
\[
  C \bullet U = (C' \otimes I)ML \bullet U = (C' \otimes I) \bullet
  ((\phi ')^{-1}(\phi(U)) = (C' \otimes I) \bullet ((\phi '')^{-1}
  \otimes I)(\phi(U)) = \phi(U)
\]
This proves existence.\\

To prove uniqueness, suppose that $D$ is another Clifford circuit in
normal form such that $D \bullet U = \phi(U)$ for all $U \in
\pauli{n}$.  By the definition of normal form, $D = (D' \otimes
I)M'L'$ where $M'$ is an $X$-circuit, $L'$ is a $Z$-circuit, and $D'$
is a normal Clifford circuit on $n-1$ qubits. Since $(D' \otimes
I)M'L' \bullet P = D \bullet P = \phi(P) = I \otimes \dots \otimes I
\otimes Z$, we have
\[
  M'L' \bullet P= (D' \otimes I)^{-1}(I \otimes \dots \otimes I
  \otimes Z) = I \otimes \dots \otimes I \otimes Z.
\]
From the uniqueness of \textbf{Proposition 4.14}, $M'=M$ and $L'=L$. Then,
by the induction hypothesis, $C'$ and $D'$ are equal up to a scalar of
$\pm 1$. Thus the same is true of $C$ and $D$. This proves uniqueness.
\end{proof}

% --------------------------------------------------------------------
\section{Proof of Corollary 4.16}
\label{app:cliffsize}

\textbf{Corollary 4.16}
  There are exactly $2 \cdot \prod_{i=1}^{n} (4^{i} + 2^{i} -2)(2
  \cdot 4^{i-1})$ real stabilizer operators on $n$ qubits.

\begin{proof}
First note that by \textbf{Definition 4.8}, the $A$ gate on the left of a
normal form will determine the input type of the first possible $B$
gate. Then the choice of each $B$ gate is dependent of the output type
of the previous gate.

There are four gates with a single input type, $B_{1}$, $B_2$, $B_3$,
and $B_{4}$, and four gates with a double input type, $B_{5}$, $B_6$,
$B_7$, and $B_{8}$. The gates $B_{1}$, $B_2$, $B_{3}$, $B_{5}$, $B_6$,
and $B_{7}$ have the output type of the top wire matching that of the
input type of the bottom wire. The gates $B_{4}$ and $B_{8}$ on the
other hand, swap between double and single types. Thus, if the last
chosen gate had a single output wire type, then we must choose one of
$B_{1}$, $B_2$, $B_3$, or $B_{4}$. Similarly, one of $B_{5}$, $B_6$,
$B_7$, or $B_{8}$ must be chosen if the previous gate had a double
output wire type.

Now to end with a circuit that is normal, the top output wire of the
last $B$ gate must be green. This means that if we start with an
$A_{1}$ gate or an $A_{2}$ gate, then we start on a green wire and we
must choose evenly many type-swapping gates ($B_{4}$ and $B_{8}$) in
our construction. Moreover, the first one of which must be a $B_{4}$
gate and the last one of which must be a $B_{8}$. If we start with an
$A_{3}$ gate, then we start on a double wire type and we must choose
oddly many type-swapping gates, the first one of which must be a
$B_{8}$ gate, and the last one of which must be a $B_{4}$ gate.

In general, the number of $Z$-circuits starting with an $A_{1}$ gate
or an $A_{2}$ gate is exactly
\[
  4 \cdot \sum\limits_{m=1}^{n} \sum\limits_{k=0}^{\lfloor
    \frac{m-1}{2} \rfloor} \binom{m-1}{2k} 3^{m-2k-1} =
  \sum\limits_{m=1}^{n} 2^{m-1}(2^{m}+2)
\]
and the number of $Z$-circuits that starting with an $A_{3}$ gate is
exactly
\[
  2 \cdot \sum\limits_{m=1}^{n} \sum\limits_{k=0}^{\lfloor
    \frac{m-1}{2} \rfloor} \binom{m-1}{2k+1} 3^{m-2(k+1)} =
  \sum\limits_{m=1}^{n} 2^{m-2}(2^{m}-2).
\] 
This produces a total of
\[
  \sum\limits_{m=1}^{n} 2^{m-1}(2^{m}+2) + \sum\limits_{m=1}^{n}
  2^{m-2}(2^{m}-2) = \sum\limits_{m=1}^{n} (2^{m-1}(2^{m}+2) +
  2^{m-2}(2^{m}-2)) = 4^{n} + 2^{n} - 2
\]
$Z$-circuits. By \textbf{Definition 4.9}, there are exactly $2 \cdot
4^{n-1}$ $X$-circuits on $n$ qubits. Since there are exactly 2
scalars, by \textbf{Definition 4.10}, there are exactly
\[
  2 \cdot \prod_{i=1}^{n} (4^{i}+2^{i}-2)(2\cdot4^{i-1})
\]
normal circuits. By Proposition \ref{autosimulate}, these are in
bijection with the elements of the $n$-qubit Clifford group.
\end{proof}
\end{document}